\documentclass[ecta,nameyear,final]{econsocart}

\RequirePackage[colorlinks,citecolor=blue,linkcolor=blue,urlcolor=blue]{hyperref}

\startlocaldefs

\usepackage{graphics}
\usepackage{amsmath, amssymb, amsthm}
\usepackage{mathpazo}
\usepackage{color}
\usepackage{mathrsfs}
\usepackage{bbm}
\usepackage[lined, ruled]{algorithm2e}
\SetAlFnt{\small\sf}
\usepackage{subfig}
\usepackage{booktabs}
\usepackage{tikz}

\DeclareMathOperator{\Var}{var}

\newcommand{\setntn}[2]{ \{ #1 : #2 \} }
\newcommand{\fore}{\therefore \quad}

\newcommand{\tod}{\stackrel { \mathscr D } {\to} }
\newcommand{\toprob}{\stackrel { p } {\to} }

\newcommand{\iidsim}{\stackrel {\textrm{ {\sc iid }}} {\sim} }
\newcommand{\1}{\mathbbm 1}
\newcommand{\la}{\langle}
\newcommand{\ra}{\rangle}

\newcommand{\given}{\, | \,}

\newcommand{\hH}{\mathscr H}

\newcommand{\eE}{\mathcal E}
\newcommand{\fF}{\mathscr F}

\newcommand{\xX}{\mathcal X}

\newcommand{\RR}{\mathbbm R}
\newcommand{\NN}{\mathbbm N}

\newcommand{\XX}{\mathbbm X}

\newcommand{\bP}{\mathbf P}
\newcommand{\bQ}{\mathbf Q}
\newcommand{\bE}{\mathbf E}

\theoremstyle{plain}

\newtheorem{theorem}{Theorem}[section]

\newtheorem{lemma}{Lemma}[section]

\theoremstyle{definition}

\newtheorem{example}{Example}[section]
\newtheorem{remark}{Remark}[section]

\newtheorem{assumption}{Assumption}[section]

\endlocaldefs

\begin{document}

\begin{frontmatter}

\title{A New Approach to Goodness of Fit for Ergodic \\ Markov Processes}
\runtitle{Goodness of Fit for Ergodic Markov Processes}

\begin{aug}
\author[add1]{\fnms{Vance}~\snm{Martin}\ead[label=e1]{vance@unimelb.edu.au}}
\author[add2]{\fnms{Yoshihiko}~\snm{Nishiyama}\ead[label=e2]{nishiyama@kier.kyoto-u.ac.jp}}
\author[add3]{\fnms{John}~\snm{Stachurski}\ead[label=e3]{j-stachurski@grips.ac.jp}}
\author[add4]{\fnms{Yiran}~\snm{Xie}\ead[label=e4]{yiran.xie@sydney.edu.au}}

\address[add1]{%
\orgdiv{Department of Economics},
\orgname{The University of Melbourne}}

\address[add2]{%
\orgdiv{Institute of Economic Research},
\orgname{Kyoto University}}

\address[add3]{%
\orgname{National Graduate Institute for Policy Studies (GRIPS)}}

\address[add4]{%
\orgdiv{School of Economics},
\orgname{University of Sydney}}
\end{aug}

\begin{funding}
The authors are grateful to Tim Kane, Peter Glynn,
Yoichi Nishiyama, participants at the 2009 NCER conference at Princeton
and the 2011 Australasian Meetings of the Econometric Society, and seminar
participants at the Empirical Micro Research seminar at Tokyo University, the
Nakanoshima Workshop at University of Osaka, and the 2011 Workshop on Statistical Analysis and Related
Topics at Tokyo University.
Our research was supported in part by Australian Research Council Grants DP120100321
and DP0987589, and by Japan Society for the Promotion of
Science Grants-in-Aid 22330067.
\end{funding}

\begin{abstract}
    We introduce a new density-based goodness of fit test for
    ergodic Markov processes.  Our test compares the data against the 
    class of models specified in the null hypothesis, and rejects if no model
    in the class yields a stationary density that matches with the data.  No
    alternative needs to be specified in order to implement the test.
    Although our test compares densities, estimation of smoothing parameters
    is not required, and the test has nontrivial power against $1/\sqrt{n}$
    local alternatives.  The test provides new perspectives on some
    existing problems in econometric and financial modeling.
\end{abstract}

\vspace{0.5em}


\end{frontmatter}

\section{Introduction}

For a dynamic stochastic model used in a given application,
an overriding concern is whether or not the dynamics of the model
are consistent with the time series being modeled.  To give one of many
possible examples, valuations of interest rate derivative securities
depend on the underlying model used to represent the interest rate.  If the
model fit is poor, in the sense that probabilities implied by the model are
inconsistent with actual interest rate dynamics, then the resulting valuation
will be unreliable.

Model evaluation through classical tests oftens require specification of the
alternative hypothesis.  One problem here is that, in many settings,
the theory says little about the set of possible alternatives.  In this case,
a natural approach is to use goodness of fit tests, where the alternative is
unspecified.  A variety of goodness of fit tests have been proposed
in the literature.  These include the nonparametric test of \cite{as96}, the Pearson, Kolmogorov-Smirnov and Cram\'er-von Mises tests and their
various extensions (see, for example, \cite{kp00,cr28,vm31,s36,d55,dur73,po84,an97,dbdg07}), and other
well-known tests such as those proposed by \cite{cw99,bai03,fz03,hl05,cgt08,gc08,ne08,asfp09,k11}.  

For goodness of fit tests of dynamic stochastic models that put little or no
structure on the alternative, one potential problem is that stochastic
processes are relatively complex objects, described by high-dimensional joint
distributions.  As a result, the power of a given test can be dispersed over a
huge space of possible alternatives.  Faced with this reality, there is often
a desire to maintain power in certain directions that are judged to be
important within the context of a given application,  while simultaneously
avoiding strict assumptions on the set of possible alternatives.

One way to achieve this goal is to select a particular implication of the
model that is considered to be important, and then test for nonspecific
departures from this restriction.  An example of this approach is the Hansen
J-test \cite{h82}, which compares the data against theoretical moment
restrictions.  Another is the test proposed by \cite{as96}, where a
nonparametric kernel density estimate of the stationary distribution is
compared to the stationary density of the model.  A\"it-Sahalia's strategy is
appropriate when correct modeling of stationary equilibrium outcomes is
important for the problem at hand, as the test concentrates power against
alternatives with stationary distributions that differ from the stationary
distribution under the null.\footnote{The Hansen J-test can also be regarded
as a test of steady state implications in the time-series setting.  The
potential benefits of matching steady state implications are discussed in
\cite{ashs}.}

Our paper also follows this approach. We propose a novel goodness of fit test
for stationary dynamic models with the Markov property.  As in \cite{as96}, our test compares the data against the set of stationary densities
corresponding to the set of models contained in the null hypothesis, and
rejects if no model in this class yields a stationary density that matches
with the data.  No alternative needs to be specified in order to implement the
test.  Unlike A\"it-Sahalia's test, however, our test is formulated for Markov
processes of arbitrary dimension, involves no smoothing parameters, and is
powerful against $1/\sqrt{n}$ local alternatives regardless of the dimension
of the state space.\footnote{For our test the state space is an
arbitrary measure space, so that, from a theoretical perspective, the dimension of
the space is completely unrestricted.
A\"it-Sahalia's test can be extended to the multivariate case, but
the performance will be compromised because nonparametric kernel density
estimators degrade rapidly as the dimension of the state space increases.}
Another difference from A\"it-Sahalia's test is that the asymptotic
distribution of our test statistic depends explicitly on the time series
properties of the process under the null hypothesis.\footnote{\cite{p98},
p.~455, pointed out that the asymptotic distribution of A\"it-Sahalia's test
statistic depends only on the stationary distribution of the process under the null,
and that this property can bias the size of the test in finite
samples.\label{fn:pr}}


\subsection{The LAE Test with a Simple Null Hypothesis}

\label{ss:sn}

In what follows, we refer to our test as the LAE test.\footnote{Here LAE
stands for ``look-ahead estimator.''  The connection between our test and the
look-ahead estimator is described in footnote~\ref{fn:lae}.}
To describe the LAE test in an abstract setting, we begin with a goodness of fit test
for a single model $p$, where $p$ is a Markov transition density kernel.
Heuristically, $p$ represents a data generating process where $p(x, y)dy$ is
the probability of transitioning from state $x \in \XX$ to state $y \in \XX$
over one unit of time.  Suppose that $p$ is ergodic with a unique stationary
density $\psi$. By definition, $\psi$ satisfies
\begin{equation}
    \label{eq:statden00}
    \int p(x, y) \psi(x) dx = \psi(y)
    \qquad (y \in \XX)
\end{equation}
Our interest is in testing whether some given $\XX$-valued time series
$\{X_t\}_{t=1}^n$ is generated by $p$.  To
test the validity of this null hypothesis, consider the deviation
\begin{equation}
    \label{eq:dev0}
    \left|
        \frac{1}{n} \sum_{t=1}^n p(X_t, y) - \psi(y)
    \right|
\end{equation}
When the null holds, the sequence
$\{X_t\}_{t=1}^n$ is stationary and ergodic with common density $\psi$, and
hence, for large $n$,
\begin{equation*}
    \frac{1}{n} \sum_{t=1}^n p(X_t, y) - \psi(y)
    \approx \bE p(X_t, y) - \psi(y)
    = \int p(x, y) \psi(x) dx - \psi(y)
\end{equation*}
By \eqref{eq:statden00} this evaluates to zero, so, under the null,
the deviation in \eqref{eq:dev0} should be small for large
$n$.  Moreover, since this argument is valid for any given $y$, we can adopt a
functional perspective, regarding 
\begin{equation}
    \label{eq:fal}
    \frac{1}{n} \sum_{t=1}^n p(X_t, \cdot) - \psi(\cdot)  
\end{equation}
as a random element taking values in the function space $L_2$,
and rejecting the null when its norm is large---that is, when its
realization lies outside a sphere $B(r_n, 0)$ centered on the origin of $L_2$.  The radius
$r_n$ of the sphere is computed from an $L_2$ central limit theorem to produce
a test of given size.

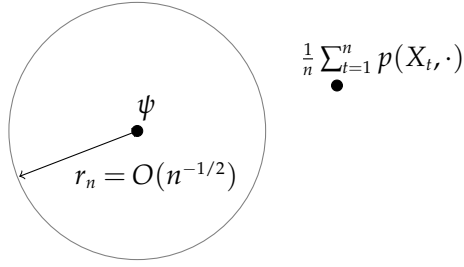
\begin{figure}
   \begin{center}
        \begin{tikzpicture}[scale=1.2]

            \node at (0.1,0) [above] {$\psi$};
            \draw[gray] (0,0) circle (1.414cm);
            \draw[fill=black] (0,0) circle (0.06cm);
            \draw[fill=black] (2.2,0.5) circle (0.06cm);
            \node at (2.7,0.5) [above] {$\frac{1}{n} \sum_{t=1}^n p(X_t, \cdot)$};
            \draw[<->] (0,0) -- (-1.3,-0.5);
            \node at (-0.8, -0.5) [right] {$r_n = O(n^{-1/2})$};

        \end{tikzpicture}
        \caption{\label{f:circles} Reject if $\frac{1}{n} \sum_{t=1}^n p(X_t, \cdot)
            \notin B(r_n, \psi)$}
   \end{center}
\end{figure}

The statement that $\frac{1}{n} \sum_{t=1}^n p(X_t, \cdot) - \psi(\cdot)$ lies
outside $B(r_n,0)$ is equivalent to the statement that $\frac{1}{n}
\sum_{t=1}^n p(X_t, \cdot)$ lies outside $B(r_n, \psi)$, the sphere of radius $r_n$ centered on
$\psi$.  This perspective is illustrated in figure~\ref{f:circles}.  The test
compares the theoretical stationary distribution $\psi$ against an estimate
$\frac{1}{n} \sum_{t=1}^n p(X_t, \cdot)$ that is $\sqrt{n}$-consistent for
$\psi$ under the null, and rejects if the deviation is greater than $r_n$.  The radius $r_n$ is shown to be of the form $c /
\sqrt{n}$, where $c$ depends on the size of the test.  The fact that the
radius is $O(1/\sqrt{n})$ suggests that the test will have nontrivial power
against $1/\sqrt{n}$ local alternatives.  This intuition is confirmed in
section~\ref{ss:la} by applying Le Cam's notion of continguity.\footnote{The
consistency of $\frac{1}{n} \sum_{t=1}^n p(X_t, \cdot)$ for $\psi$ has also
been studied in a computational (rather than statistical) setting.  The idea
is that if $\psi$ is intractable but $p$ is known and $\{X_t\}$ can be
simulated from $p$, then $\frac{1}{n} \sum_{t=1}^n p(X_t, \cdot)$ can be used
as an approximation of $\psi$.  In this setting, $\frac{1}{n} \sum_{t=1}^n
p(X_t, \cdot)$ is called the look-ahead estimator of $\psi$.  See
\cite{gh98}.\label{fn:lae}}

\subsection{The LAE Test with Estimated Parameters}

\label{ss:epd}

The goodness of fit test described above is mainly of theoretical interest.
In practical situations, models usually contain unknown parameters,
and we wish to test whether our parametric class of models can represent
the data.  We extend to this setting by taking a parametric family of Markov
models, indexed by a vector $\theta \in \Theta$.  In particular, for each
$\theta$, we take $p(\theta,x,y)$ to be a density transition kernel, and let
$\psi(\theta,y)$ be the corresponding stationary density.  The null hypothesis
is that $\{X_t\}_{t=1}^n$ is generated by $p(\theta,x,y)$ for some $\theta
\in \Theta$.  Taking $\{\hat \theta_n\}$ to be a $\sqrt{n}$-consistent
estimator of the true parameter vector under the null, we propose a test statistic
based on the $L_2$ norm of
\begin{equation*}
    \frac{1}{n} \sum_{t=1}^n p(\hat \theta_n, X_t, \cdot) - \psi(\hat \theta_n, \cdot)
\end{equation*}
This expression is the parametric analogue of \eqref{eq:fal}, and similar
intuition applies.  However, the asymptotic theory must be adjusted to take
into account randomness in parameter estimates.  See section~\ref{s:ep} for
details.

\subsection{Other Tests}

\label{ss:ot}

The LAE test we propose in this paper is closely related to three important tests.
One of these is the Cram\'er-von Mises test \cite{cr28,vm31},
which was generalized to the time series setting by \cite{as93,cb}.  To illustrate the connection
between the Cram\'er-von Mises test and the LAE test, consider a simple case
without estimated parameters, and where the state space $\XX$ is the $[0,1]$
inverval.  (More general cases change little in what follows.) Let
$\{X_t\}_{t=1}^n$ be an $\XX$-valued process, and let $\Psi_n(y) :=
n^{-1} \sum_{t=1}^n \1\{X_t \leq y\}$ be the empirical distribution of the
data.  Take as the null hypothesis the statement that $\{X_t\}_{t=1}^n$ is
stationary and geometrically ergodic with common cumulative distribution
function $\Psi$.  Under the null, it can be shown that 
\begin{equation}
    \label{eq:cvmt}
    n \| \Psi_n - \Psi \|^2
    \; \tod \;
    \sum_{\ell = 1}^{\infty} \sigma_{\ell} Z_{\ell}^2 
    \qquad \text{as} \quad
    n \to \infty
\end{equation}
where $\| \cdot \|$ is the $L_2$ norm on $\XX$, $\{Z_{\ell}\}$ is an {\sc iid} sequence
of scalar standard normal random variables and $\{\sigma_{\ell}\}_{\ell \geq
1}$ are the eigenvalues of a certain covariance operator $\Sigma$ (cf., e.g., \cite{dbdg07}).  Taking
  $c_{\alpha}$ to be the $1-\alpha$ quantile of the random variable on the
  right-hand side of \eqref{eq:cvmt}, a test rejecting the null hypothesis
  whenever $n \| \Psi_n - \Psi \|^2 > c_{\alpha}$ is asymptotically of size
  $\alpha$.

As can be seen by referring back to section~\ref{ss:sn}, the LAE test has a
very similar structure to the Cram\'er-von Mises test, with the primary
difference being that the LAE test statistic looks at $L_2$ deviations between
densities rather than distribution functions.  In particular, in the LAE test,
\begin{enumerate}
    \item the cdf $\Psi$ is replaced by the corresponding density $\psi$,
    \item the cdf estimate $\Psi_n$ is replaced by the density estimate 
        $n^{-1} \sum_{t=1}^n p(X_t, \cdot)$ on the
        left-hand side of \eqref{eq:fal}, and
    \item the covariance operator $\Sigma$ is altered by these modifications,
        leading to different eigenvalues and hence a different critical value
        $c_{\alpha}$.
\end{enumerate}

One advantage of the LAE test over the Cram\'er-von Mises test is that in the 
Cram\'er-von Mises test, the rejection
criterion is based on the $L_2$ deviation between estimated and hypothesized
cdfs. As a measure of deviation between two distributions, $L_2$ deviation
between the cdfs is relatively difficult to motivate in an econometric
setting.  On the other hand, the LAE test is concerned with $L_2$ deviation
between \emph{densities}.  This measure is straightforward to motivate.  For
example, consider a setting where an agent chooses an optimal action $a^*$ by
minimizing expected loss $\ell(a, \psi) := \int L(a, y) \psi(y) dy$.  Here
$\psi$ is the density of a vector of relevant state variables, and $L(a,y)$ is
the loss from choosing action $a$ when the realized state is $y$ (see, e.g.,
\cite{dgl98}).  Suppose we have a model that implies some
density $\psi_m$ for the state.  The true and unknown density for the state we
denote by $\psi_0$.  In this setting, we wish to know whether $\min_a \ell(a,
\psi_m)$ is close to the true minimum $\min_a \ell(a, \psi_0)$.  Bounding the
deviation between these minima requires a bound on $\sup_a |\ell(a,\psi_m) -
\ell(a, \psi_0)|$.  Such a bound can be obtained via the Cauchy-Schwartz
inequality, which yields
\begin{equation*}
    \sup_a |\ell(a, \psi_m) - \ell(a, \psi_0)| 
        \leq  \sup_a \left[ \int L(a,y)^2 dy \right]^{1/2}
        \left[ \int (\psi_m(y) - \psi_0(y))^2 dy \right]^{1/2}
\end{equation*}
The term on the far right is the $L_2$ deviation between the densities $\psi_m$ and $\psi_0$.

The LAE test is not the first to consider a test
statistic based on $L_2$ deviation between densities.  A well-known test based
on this deviation was constructed by \cite{as96}. The main difference
between the LAE test and the one proposed by A\"it-Sahalia is that in the
latter, the $L_2$ deviation is between the hypothesized density and a
nonparametric kernel density estimate constructed from the data.  The
relationship between the LAE test and A\"it-Sahalia's test was discussed
earlier in the introduction.  Further discussion of the relationship between
the tests is given in section~\ref{s:d}.

Finally, we mention a third test that is closely connected to the one proposed
in this paper.  Recall that Hansen's \cite{h82} J-test begins with a moment
restriction of the form $\bE g(X_t, \theta) = 0$ for some function $g$.  The
null hypothesis of the test is 
\begin{equation}
    \label{eq:gmmnull}
    H_0 \colon \; \exists \, \theta \in \Theta \; 
        \text{ such that } \; \bE g(X_t, \theta) = 0
\end{equation}
The null hypothesis is rejected if 
\begin{equation}
    \label{eq:hsts}
    n \left\| \frac{1}{n} \sum_{t=1}^n g(X_t, \hat \theta_n) \right\|^2_W
\end{equation}
is large relative to a particular $\chi^2$ distribution, where $\| \cdot \|_W$
is a weighted euclidean norm.  

To formulate the LAE test in a parallel manner, let
\begin{equation}
    \label{eq:pbardef}
    \bar p(\theta, x, y) := p(\theta, x, y) - \psi(\theta, y) 
    \qquad (\theta \in \Theta, \;\; (x,y) \in \XX \times \XX)
\end{equation}
As in section~\ref{ss:epd} above, the
null hypothesis is that $\{X_t\}_{t=1}^n$ is generated by $p(\theta,x,y)$ for
some $\theta \in \Theta$.  Under the null, then, there exists a $\theta$ with
$X_t \sim \psi(\theta,\cdot)$ for all $t$, and hence \eqref{eq:statden00} implies that 
\begin{equation*}
    \bE \bar p(\theta,X_t,y)
    = \bE p(\theta,X_t,y) - \psi(\theta, y) 
    = 0  
\end{equation*}
Treating all $y$ simultaneously, we can write this restriction as 
\begin{equation}
    \label{eq:gmmfh}
    \exists \, \theta \in \Theta \; \text{ such that } \; 
    \eE \bar p(\theta, X_t, \cdot) = 0
\end{equation}
where $\eE$ is a functional expectation for random elements of $L_2$, and the
zero on the right-hand side is the origin of $L_2$.  This is an
infinite-dimensional version of \eqref{eq:gmmnull}, and the LAE test statistic is
analogous to \eqref{eq:hsts} when $g$ is replaced by $\bar p$ and the norm in \eqref{eq:hsts} is replaced
with the $L_2$ norm.\footnote{An infinite-dimensional Hansen J-test was
considered in the elegant paper of \cite{cf00}.
Their focus is mainly on estimation with {\sc iid} observations.  (Ours is on
testing with dependent
observations.) Their theoretical results on the Hansen J-test cannot be
used here, since we permit the state space for the Markov process to be
multidimensional (it may in fact be infinite-dimensional), we permit the
parameters to be estimated by any $\sqrt{n}$-consistent estimator, and, in
addition, our data are explicitly Markovian under the null.  Since not all
ergodic Markov processes satisfy central limit theorems, a careful asymptotic
theory is provided.}

\section{Preliminaries}

\label{s:prel}

We begin by recalling some elementary facts regarding random variables in a separable Hilbert
space $\hH$ with inner product $\la \cdot, \cdot \ra$.  An \emph{$\hH$-valued
random variable} $F$ on probability space $(\Omega, \fF, \bP)$ is a measurable
map from $(\Omega, \fF)$ into $\hH$ paired with its Borel sets.  If $\bE \|F\|
< \infty$, where $\bE$ is the ordinary scalar expectation, then the
\emph{vector (or Pettis) expectation} $\eE F$ of $F$ is the unique element of
$\hH$ satisfying $\la \eE F, h \ra = \bE \la F, h \ra$ for all $h \in
\hH$. In most of what follows, $\hH$ will be a 
function space, and the vector expectation of a random function can be
obtained by taking scalar expectations at each point: If $F$
is a random function of the form $F(y) = f(X,y)$ where $X$ is some random
variable and $f$ is real-valued, then $\eE F$ is given by the
function $m_f(y) := \bE f(X,y)$.\footnote{For more
details on the Pettis expectation, see, for example, \cite{b00}, chapter~1.} 

If $\bE \| F \|^2 < \infty$,
then the \emph{covariance operator} $C$ of $F$ is the linear operator defined
by
\begin{equation}
    \label{eq:defco}
    \la g, C h \ra = \bE \la g, F - \eE F \ra \la h, F - \eE F \ra
    \qquad \forall \; g, h \in \hH
\end{equation}
An $\hH$-valued random variable $G$ is called \emph{Gaussian} if $\la h, G \ra$ is
normally distributed on $\RR$ for every $h \in \hH$.   We write $G \sim
N(m, C)$ if $G$ is Gaussian on $\hH$ with mean $m = \eE G$ and covariance
operator $C$.  Letting $\{Z_{\ell}\}_{\ell \geq 1}$ be an {\sc iid} sequence
of standard normal random variables and $\{\lambda_{\ell}\}_{\ell \geq 1}$ be
the eigenvalues of $C$, we can characterize the distribution of $\|G\|^2$ as
follows:

\begin{lemma}
    \label{l:dsn}
    If $G \sim N(0, C)$ on $\hH$, then $\| G \|^2$ has the same distribution as 
        $\sum_{\ell = 1}^{\infty} \lambda_{\ell} Z_{\ell}^2$.
\end{lemma}

\subsection{Set Up}

\label{ss:su}

We consider stochastic processes taking values in an arbitrary state space $\XX$, with
countably generated $\sigma$-algebra $\xX$ and $\sigma$-finite measure $\mu
\colon \xX \to \RR_+$. To simplify notation, we use symbols such as $dx$ and
$dy$ to indicate integration with respect to $\mu$, rather than $\mu(dx)$ and
$\mu(dy)$.  Two common settings are where
\begin{enumerate}
    \item $\XX$ is a Borel subset of $\RR^k$ and $\mu$ is Lebesgue measure.
    \item $\XX$ is discrete and $\mu$ is the counting measure.
\end{enumerate}
A density on $\XX$ is any $\xX$-measurable $f \colon
\XX \to \RR_+$ with $\int f(x) \, dx = 1$.
A \emph{density kernel} on $\XX$ is an $\xX \otimes \xX$-measurable
function $p \colon \XX \times \XX \to \RR_+$ such that $p(x,\cdot)$ is a
density on $\XX$ for all $x \in \XX$.  An $\XX$-valued process
$\{X_t\}_{t \in \NN}$ will be called \emph{$p$-Markov} if it is stationary, 
Markov, and $p(X_t, \cdot)$
is the conditional density of $X_{t+1}$ given $X_t$ for all $t$.

\begin{example}
    \label{eg:v}
    Let $\XX = \RR$, let $\xX$ be the Borel sets and let $\mu$ be Lebesgue
    measure.  Under the \cite{v77} model, the rate of interest $X_t$ follows
    \begin{equation}
        \label{eq:v}
        d X_t = \kappa (b - X_t) d t + \sigma d W_t
    \end{equation}
    where $\kappa$, $b$ and $\sigma$ are parameters, and $W_t$ is Brownian
    motion.  The transition probability function associated with this process
    is 
    \begin{equation}
      \label{eq:vk}
      q(t, x, y) := 
      \{2\pi v(t)\}^{-1/2}
      \exp\left\{ \frac{- (y - m(t, x))^2 }{2 v(t)}
                  \right\}
    \end{equation}
    where $v(t) := \sigma^2 (1 - e^{-2\kappa t}) / (2 \kappa)$ and 
    $m(t, x) := b + (x - b) e^{-\kappa t}$.  If a unit of time
    corresponds to one year and $\{X_t\}_{t=1}^n$ is a sequence of
    monthly observations from the process (\ref{eq:v}), then $\{X_t\}_{t=1}^n$
    is $p$-Markov for $p(x,y) := q(1/12, x, y)$.
\end{example}

\begin{example}
    \label{eg:ar1}
    Let $\XX = \RR^k$, let $\xX$ be the Borel sets, and let $\mu$ be Lebesgue measure. 
    Consider a stationary nonlinear AR(1) process
    \begin{equation}
        \label{eq:srm}
        X_{t+1} = g(X_t) + W_{t+1} 
        \qquad (W_t)_{t\geq 1} \iidsim \phi
    \end{equation}
    where $\phi$ is a density on $\RR^k$ and $g$ is a measurable function from
    $\RR^k$ to itself.  The sequence $\{X_t\}$ in (\ref{eq:srm}) is $p$-Markov for
    \begin{equation}
        \label{eq:srmk}
        p(x,y) := \phi(y - g(x))
        \qquad ( (x,y) \in \RR^k \times \RR^k)
    \end{equation}
\end{example}

\begin{example}
    \label{eg:disc}
    Let $\XX = \{1,\ldots,N\}$, let $p$ be a stochastic $N
    \times N$ matrix,\footnote{That is, $p(x,y) \geq 0$ for each $(x,y) \in
    \XX \times \XX$, and $\sum_{y\in \XX} p(x,y) =1$ for each $x \in \XX$.}
    and let $\{X_t\}$ be a stationary Markov chain on $\XX$ satisfying 
    \begin{equation*}
        \bP\{X_{t+1} = y \,|\, X_t = x\} = p(x, y)   
        \qquad ((x,y) \in \XX \times \XX)
    \end{equation*}
    If $\xX := \setntn{B}{B \subset \XX}$ and $\mu$ is the counting measure,
    then $p$ is a density kernel on $\XX$, and $\{X_t\}$ is
    $p$-Markov.
\end{example}

Returning to the general case, let a density kernel $p$ be given, and consider a
$p$-Markov process $\{X_t\}$ on $\XX$.  The conditional distribution of $X_t$
given $X_0=x$ is represented by the $t$-th order density $p^t(x,\cdot)$, where
$p^1 := p$ and
\begin{equation*}
    p^t(x,y) := \int p(x,z) p^{t-1} (z,y) dz
    \qquad ( (x,y) \in \XX \times \XX)
\end{equation*}
A density $\psi$ on $\XX$ is called \emph{stationary} with respect to $p$ if 
\eqref{eq:statden00} holds.  In all cases we consider, $p$ will have a unique
stationary density $\psi$.  If $\{X_t\}$ is stationary and $p$-Markov, then
$X_t \sim \psi$ for all $t \geq 0$.  To simplify notation, in what follows we
let
\begin{equation}
    \label{eq:pbardef2}
    \bar p(x,y) := p(x,y) - \psi(y)
    \qquad ((x,y) \in \XX \times \XX)
\end{equation}
Let $L_2 := L_2(\XX, \xX, \mu)$ denote the set of all $\xX$-measurable
functions $h$ mapping $\XX$ to $\RR$ such that $\int h(x)^2 dx := \int h(x)^2
\mu(dx)$ is finite.  As usual, elements of $L_2$ equal $\mu$-almost everywhere
are identified.  The inner product and norm on $L_2$ are defined by $\la g, h
\ra := \int g(x) h(x) dx$ and $\|h\| := \la h, h \ra^{1/2}$.  Since we have
assumed that $\xX$ is countably generated, the Hilbert space $(L_2, \| \cdot \|)$ is
separable.

\subsection{Ergodicity Assumptions}

\label{ss:ea}

Our asymptotic theory relies on a central limit theorem, which in turn depends on the
properties of the underlying Markov process under the null hypothesis.  We
will assume that the process
is geometrically ergodic. For a given density kernel $p$, ergodicity requires
that $p$ has a unique stationary density $\psi$, and any $p$-Markov process
satisfies the strong law of large nameyear (see \cite{mt09}, theorem
17.1.7, and \cite{l02}, theorem~21.12 for the many equivalent
definitions of ergodicity).  Geometric ergodicity requires that, in
addition, there exist positive constants $\lambda < 1$ and $L < \infty$ and 
a weight function $V \colon \XX \to \RR_+$ such that
\begin{equation}
    \label{eq:vuedef}
    \int V(x) \psi(x) dx < \infty
    \quad \text{and} \quad
    \left| \int_B p^t(x,y) dy - \int_B \psi(y)dy \right|
    \leq \lambda^t L V(x)
\end{equation}
for all $B \in \xX$, $x \in \XX$ and $t \in \NN$.\footnote{\cite{k06}
gives geometric ergodicity conditions for a number of
popular time-series models, including (nonlinear) ARMA, bilinear, GARCH and
random coefficient models.  \cite{ns04} demonstrate
geometric ergodicity of the one-sector stochastic optimal growth model under
the classical assumptions.  \cite{mt09}, chapter~15, provide a
general treatment.} In what follows, we will say that $p$ is \emph{$V$-mixing} if
there exists a function $V \colon \XX \to \RR_+$ such that $p$ is
geometrically ergodic with weight function $V$, and, in addition, there are
nonnegative constants $c_0$, $c_1$ and $\gamma$ with $\gamma < 1$ and
\begin{equation}
    \label{eq:bop}
    \int p(x, y)^2 dy \leq c_0 + c_1 V(x)^{\gamma} \qquad \forall \, x \in \XX
\end{equation}
Together, geometric ergodicity and \eqref{eq:bop} provide the mixing and moment
conditions necessary for our asymptotic theory to hold.

\begin{example}
    The Vasicek density kernel $p$ in example~\ref{eg:v} is $V$-mixing
    with $V(x) = |x|$ whenever $\kappa > 0$.  The unique stationary density
    $\psi$ is $N(b, \sigma^2/(2\kappa))$.  
\end{example}

\begin{example}
    Consider $p$ in example~\ref{eg:ar1}. Let $\|\cdot\|_{E}$ be the euclidean
    norm on $\RR^k$.  If $g$ and $\phi$ are both continuous, $\phi$ is
    strictly positive on $\RR^k$, $\int
    \phi(z)^2 dz < \infty$, and there exist constants $\alpha \in [0,1)$
    and $\beta \in \RR_+$ such that $\|g(x)\|_E \leq \alpha \|x\|_E + \beta$
    for all $x \in \RR^k$, then $p$ is $V$-mixing with $V(x) := \|x\|_E$.\footnote{By \cite{mt09},
    prop~6.1.5, thm.~6.2.9, and thm.~16.1.2, $p$
    is geometrically ergodic with weight function $V(x) := \|x\|_E$.
    Moreover, \eqref{eq:bop} is satisfied with $c_0 := \int
    \phi(y)^2 dy$ and $c_1 = 0$.}
\end{example}

\begin{example}
    If the discrete density kernel $p$ in example~\ref{eg:disc} is irreducible and
    aperiodic, then $p$ is $V$-mixing with $V \equiv 1$.
\end{example}


\section{Goodness of Fit for Markov Processes}
\label{s:npgft}

We begin discussion of the LAE test in this section by looking at a simple null
hypothesis, corresponding to the statement that the data are generated by a
particular density kernel $p$.  (The composite null case is treated from
section~\ref{s:ep} on.)  Our asymptotic theory for the simple null case will
make use of the following result, which is similar to theorem~1 of \cite{sm08}.  In the statement of the result, the symbol $N(0,\Lambda)$ represents a
centered Gaussian distribution on $L_2$ with covariance operator $\Lambda$,
convergence in distribution has the usual meaning,\footnote{If $E$ is a metric space, then a sequence of
$E$-valued random variables $\{Y_n\}$ converges in distribution to an
$E$-valued random variable $Y$ if $\bE g(Y_n) \to g(Y)$ for every
continuous bounded $g \colon E \to \RR$.}
and $\bar p$ is defined in \eqref{eq:pbardef2}.

\begin{theorem}
    \label{t:bk2}
    If $p$ is $V$-mixing and $\{X_t\}$ is $p$-Markov,
    then
    \begin{equation}
        \label{eq:far}
        n^{-1/2} \sum_{t=1}^n \bar p(X_t, \cdot) \tod N(0,\Lambda)
        \qquad (n \to \infty)
    \end{equation}
    for the covariance operator $\Lambda$ satisfying 
    \begin{equation}
        \label{eq:defc0}
        \la h, \Lambda h \ra 
        = \bE \la \bar p(X_1, \cdot), h \ra^2
        + 2 \sum_{t=2}^{\infty} \bE  \la \bar p(X_1, \cdot), h \ra
                \la \bar p(X_t, \cdot), h \ra
                \qquad (h \in L_2)
    \end{equation}
\end{theorem}

The proof of theorem~\ref{t:bk2} can be found in section~\ref{s:p}.  
From theorem~\ref{t:bk2}, lemma~\ref{l:dsn} and the
continuous mapping theorem, it follows that if $\{Z_{\ell}\}_{\ell \geq 1}$ is
an {\sc iid} sequence of scalar standard normal random variables and
$\{\lambda_{\ell}\}_{\ell \geq 1}$ are the eigenvalues of $\Lambda$, then
\begin{equation}
    \label{eq:adts}
    T_n :=
    \frac{1}{n} 
    \int \left\{ \sum_{t=1}^n \bar p(X_t, y) \right\}^2 dy 
    \; \tod \;
    \sum_{\ell = 1}^{\infty} \lambda_{\ell} Z_{\ell}^2 
    \qquad (n \to \infty)
\end{equation}
Take $p$ to be a fixed $V$-mixing density kernel, and
suppose that we have $n$ observations of an $\XX$-valued
stochastic process $\{X_t\}$.  Under the null hypothesis
that $\{X_t\}_{t=1}^n$ is $p$-Markov,    
the convergence in \eqref{eq:adts} is valid, so if $\alpha \in
(0,1)$ and $c^{\Lambda}_{\alpha}$ is the $1-\alpha$ quantile of $\sum_{\ell}
\lambda_{\ell} Z_{\ell}^2$, then the test
\begin{equation}
    \label{eq:test}
    \text{reject $H_0$  if $T_n > c^{\Lambda}_{\alpha}$}
\end{equation}
is asymptotically of size $\alpha$.  The integral in the definition of $T_n$
can be computed numerically.  Computation of $c^{\Lambda}_{\alpha}$ is
discussed in section~\ref{ss:ccv1}.

\begin{remark}
    In essence, \eqref{eq:test} is a test of stationary outcomes, and power is
    concentrated against alternatives with stationary distributions that
    differ significantly from $\psi$.  The converse to this property is that
    the test is not consistent against alternatives $p_1 \not= p$ that share
    the same stationary density as the null.  On the other hand, if we
    rephrase our null hypothesis in line with the infinite dimensional J-test
    interpretation (see section~\ref{ss:ot}), consistency is
    recovered. A statement of this result in a more general setting is provided in
    section~\ref{ss:ct}.
\end{remark}


\begin{remark}
    Unlike the Cram\'er-von Mises test, \eqref{eq:test} has no obvious
    equivalent in the {\sc iid} case.  Although the test remains 
    well-defined with an {\sc iid} null, such a null implies that $p(x,\cdot)
    = \psi(\cdot)$ for all $x$, or $\bar p = 0$. When $\bar p = 0$, the test
    statistic in \eqref{eq:test} is identically zero.
\end{remark}

%
%

\subsection{Computing Critical Values}

\label{ss:ccv1}

As shown in \eqref{eq:adts}, the asymptotic distribution of the test
statistic $T_n$ under $H_0$ depends on $\{\lambda_{\ell}\}_{\ell \geq 1}$, the eigenvalues of
the covariance operator $\Lambda$ in \eqref{eq:defc0}.  It can be shown that
$\Lambda$ has the integral representation
    $\Lambda h(y') := \int \zeta (y,y') h(y) dy$,
where the \emph{covariance function} $\zeta$ is given by\footnote{In
\eqref{eq:srvcf}, the expression $\bar p^t$ is defined by $\bar p^t(x, y) :=
p^t(x, y) - \psi(y)$.}
\begin{multline}
    \label{eq:srvcf}
  \zeta(y,y') 
  := \int \bar p(x,y)\bar p(x,y') \psi(x) dx  
  \\
       + \sum_{t = 2}^{\infty} 
       \left\{
        \int \bar p(x,y)\bar p^t(x,y') \psi(x) dx 
        + 
        \int \bar p(x,y')\bar p^t(x,y) \psi(x) dx 
        \right\}
\end{multline}
A plot of $\zeta$ is given in figure~\ref{f:cplot}, corresponding
to the Vasicek density kernel 
\begin{equation}
    \label{eq:blv}
    p(x,y) := q(1/12, x, y), 
    \quad \kappa = 0.85837, \; 
    b = 0.089102, \; 
    \sigma^2 = 0.0021854
\end{equation}
where $q$ is defined in \eqref{eq:vk}.  These parameter values are estimated
from US short rate data by \cite{as96}. 

\begin{figure}
  \begin{center}
    \rotatebox{0}{\scalebox{0.8}{\includegraphics{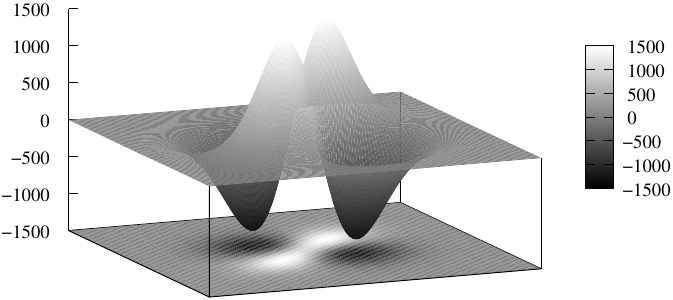}}}
  \end{center}
  \caption{\label{f:cplot} The function $\zeta(y,y')$ for the Vasicek model}
\end{figure}

To compute the critical value $c^{\Lambda}_{\alpha}$, one possibility is to approximate
the eigenvalues $\{\lambda_{\ell}\}_{\ell \geq 1}$ of $\Lambda$ using the
covariance function $\zeta$ paired with a
numerical technique such as Galerkin projection, and then take the $1-\alpha$
quantile of $\sum_\ell \lambda_{\ell} Z_{\ell}^2$.
Alternatively, $c^{\Lambda}_{\alpha}$ can be approximated by simulating the test
statistic $T_n$ under the null, as in algorithm~\ref{al:cc}.
A size-adjusted test can be produced by setting $N$ in
algorithm~\ref{al:cc} equal to the data size $n$.  In this case the
observations $\{T^m_n\}_{m=1}^M$ produced by the algorithm are {\sc iid} draws from the
distribution of $T_n$ under the null hypothesis, and hence
the $1-\alpha$ empirical quantile of these observations converges in
probability to the $1-\alpha$ quantile of the distribution of $T_n$ as $M \to \infty$.

\begin{algorithm}
    \vspace{0.6em}
    fix integers $M$, $N$\;
    \For{$m \in \{1,\dots,M\}$} { 
      \vspace{.2em}
      simulate a $p$-Markov time series $X_1^*,\ldots,X_N^*$ \;
      set $T^m_N \leftarrow N^{-1}
      \int \{ \sum_{t=1}^N \bar p(X_t^*, y) \}^2 dy$ \; 
    }
    {\bf return} the $1 - \alpha$ quantile of $T^1_N,\ldots,T^M_N$ \;
    \vspace{0.6em}
    \caption{\label{al:cc} Approximates $c^{\Lambda}_{\alpha}$ corresponding to given
    size $\alpha$ and kernel $p$}
\end{algorithm}



\subsection{Local Alternatives}

\label{ss:la}

In this section we investigate the power of the LAE test \eqref{eq:test}
against $1 / \sqrt{n}$ local alternatives.  In particular, the test we
consider is
\begin{equation*}
    H_0 \; \colon \{X_t\}_{t=1}^n \text{ is $p$-Markov}    
    \qquad \text{vs} \qquad
    H_1 \; \colon \{X_t\}_{t=1}^n \text{ is $p_n$-Markov for all } n
\end{equation*}
where $p$ is $V$-mixing, and $\{ p_n \}$ is the sequence of kernels 
defined by
$p_n(x,y) := p(x,y) + k(x,y) / \sqrt n$
for some given $k \colon \XX \times \XX \to \RR$.  To ensure that $p_n$ is
a density kernel, we require $\int k(x,y) dy = 0$ for all $x$.  We set 
\begin{equation*}
    Y_n(y) := n^{-1/2} \sum_{t=1}^n \bar p(X_t, y)
    \qquad (y \in \XX)
\end{equation*}
Thus, $Y_n$ is the random element of $L_2$ in \eqref{eq:far}, and $T_n =
\|Y_n\|^2$, where $T_n$ is the test statistic in \eqref{eq:test}.
Also, let $\tau$ be the element of $L_2$ defined by 
\begin{equation*}
    \tau(y) := \sum_{t=1}^{\infty}
    \bE \left\{ 
    \bar p(X_{t+1}, y) 
        \frac{k(X_1, X_2)}{p(X_1, X_2)}
        \right\}
\end{equation*}
where the expectation is taken under $H_0$.

\begin{assumption}
    \label{a:laa}
    Together, $k$ and $p$ satisfy the third moment condition
    \begin{equation*}
        \bE \,
        \sup_{\delta \in [0,1]} 
        \frac{ |k(X_1, X_2)|^3 }{ |p(X_1, X_2) + \delta \, k(X_1, X_2)|^3 }
        < \infty
    \end{equation*}
\end{assumption}


\begin{theorem}
    \label{t:la}
    If $H_1$ and assumption~\ref{a:laa} both hold, then $Y_n \tod N(\tau,
    \Lambda)$.
\end{theorem}

While theorem~\ref{t:bk2} showed that under $H_0$ the sequence
$\{Y_n\}$ converges in distribution to $N(0, \Lambda)$, theorem~\ref{t:la} tells us
that under $H_1$ it converges instead to $N(\tau, \Lambda)$.
Theorem~\ref{t:la} implies non-trivial power for the test \eqref{eq:test}
whenever $\tau \not= 0$, since the test statistic $T_n$ is equal to
the squared norm of $Y_n$.

In assumption~\ref{a:laa} and in the definition of $\tau$ in
theorem~\ref{t:la}, the expectation is taken under $H_0$.
The exact meaning of the claim in theorem~\ref{t:la} can be clarified as
follows:  Let $(\Omega_n, \fF_n)$ be the product space $\XX^n :=
\times_{t=1}^n \XX$ with its product $\sigma$-algebra, let $X_t \colon
\Omega_n \to \XX$ be the projection $X_t(x_1,\ldots,x_n) = x_t$, let $\bP_n$
be the distribution of $(X_1, \ldots, X_n)$ over $\XX^n$ constructed from $p$
in $H_0$, and let $\bQ_n$ be the distribution on $\XX^n$ constructed from the
local alternative $p_n$.  (Construction of $\bP_n$ and $\bQ_n$ from their
respective kernels is via the standard definition---see, e.g., \cite{mt09}, ch.~3.) The claim in theorem~\ref{t:la} is that, for all
continuous bounded $g \colon L_2 \to \RR$, we have $\int g(Y_n) d \bQ_n \to
\int g d \nu$ as $n \to \infty$, where $\nu$ is the $L_2$ Gaussian $N(\tau,
\Lambda)$.  Our proof of theorem~\ref{t:la} uses a contiguity argument, based on an
Hilbert space extension of Le Cam's third lemma.  Details are given in
section~\ref{s:p}.

\subsection{Simulation of $\psi$}

\label{ss:sop}

In applications, the stationary density $\psi$ that forms part of the test
statistic (\ref{eq:adts}) may be intractable.  In this case, one possibility
is to approximate $\psi$ via simulation.  To implement this idea, consider
again the setting of theorem~\ref{t:bk2}.  Fix $k \in \NN$, and let
$\{X_t^*\}_{t=1}^{kn}$ be a simulated $p$-Markov sequence that is independent
of the data $\{X_t\}_{t=1}^n$.  For each $k \in \NN$ we have the following
result: 

\begin{theorem}
    \label{t:dms}
    Let $\{Z_{\ell}\}_{\ell \geq 1}$ be an {\sc iid} sequence of standard normal
    random variables.  If the conditions of theorem~\ref{t:bk2} hold, then,
    as $n \to \infty$,
    \begin{equation}
        \label{eq:dmsn}
        \frac{1}{n} \int \left\{ 
        \sum_{t=1}^n p(X_t,y) - \frac{1}{k} \sum_{t=1}^{nk} p(X_t^*,y)  
                         \right\}^2 dy 
                         \tod
        (1 + 1/k) \sum_{\ell=1}^{\infty} \lambda_{\ell} Z_{\ell}^2 
    \end{equation}
\end{theorem}

Comparing \eqref{eq:dmsn} with \eqref{eq:adts}, the limit $(1 + 1/k) \sum_{\ell=1}^{\infty} \lambda_{\ell}
Z_{\ell}^2$ of the simulation-based test statistic converges almost surely to
that of the original test statistic $T_n$ as $k \to \infty$.


\section{A Specification Test for Parametric Classes}

\label{s:ep}

The LAE test in \eqref{eq:test} corresponds to the simple null that $\{X_t\}$ is
$p$-Markov, and represents a goodness of fit test for individual models.  A
more practical setting is where we have a parametric class of models, and 
aim to test the hypothesis that the data are generated by one of the models in this
class.  In this case we need to augment our asymptotic theory to accommodate
estimated parameters.

\subsection{The Test with Estimated Parameters}

Let $\Theta$ be an open convex subset of $\RR^M$, and let
$\{p_{\theta}\}_{\theta \in \Theta}$ be a parametric family of density kernels
such that $p_{\theta}$ is $V_{\theta}$-mixing for each $\theta \in \Theta$.
Let $\psi_{\theta}$ be the unique stationary density corresponding to
$p_{\theta}$.  When convenient, we write $p(\theta, x, y)$ instead of
$p_{\theta}(x,y)$, and $\psi(\theta, y)$ in place of $\psi_{\theta}(y)$.  In
addition, let $\bar p(\theta, x, y)$ be defined by \eqref{eq:pbardef}.
We begin with a limit theorem for the distribution of the test statistic in
the estimated parameter case.  In the assumptions below, $\|\cdot\|_E$ denotes
the Euclidean norm in $\RR^M$, as opposed to $\|\cdot\|$, the norm in
$L_2$.\footnote{While the following assumptions are imposed over the whole parameter space
$\Theta$, for the asymptotic theory below it would suffice that they hold over
an open neighborhood of the true parameter vector.}

We suppose the existence of an asymptotically linear and
$\sqrt{n}$-consistent sequence of estimators $\{\hat \theta_n\}$ for the
parameter vector $\theta$. 
In particular, we assume the following.

\begin{assumption}
    \label{a:le}
    There exists an $r \in \NN$ and influence function
    $g_{\theta} \colon \RR^{r+1} \to \RR^M$ such that if 
    $\{X_t\}$ is $p_{\theta}$-Markov, then
    \begin{enumerate}
        \item $\bE g_{\theta}(X_t, \ldots, X_{t+r}) = 0$
        \item $\sqrt n (\hat \theta_n - \theta) = n^{-1/2} \sum_{t=1}^{n-r}
            g_{\theta}(X_t, \ldots, X_{t+r}) + o_P(1)$.
        \item $\| g_{\theta}(x_0, \ldots, x_r) \|_E^{2+\delta} \leq 
            \sum_{k=0}^r V_{\theta} (x_k)$ on $\XX^{r+1}$
            for some $\delta > 0$
    \end{enumerate}
\end{assumption}

These assumptions imply that if $\{X_t\}$ is
$p_{\theta}$-Markov, then $\sqrt{n} (\hat \theta_n - \theta) = O_P(1)$.

\begin{assumption}
    \label{a:coid0}
    The vector $D \bar p(\theta, x, y)$ of partial derivatives $\frac{\partial}{\partial
    \theta_m} \bar p(\theta, x, y)$ exists and satisfies
    \begin{equation*}
        \int \left\{ 
            \frac{\partial}{\partial \theta_m} \bar p(\theta, x, y)
             \right\}^2
        dy \leq V_{\theta}(x)^{1/2}
        \quad
        \text{ for all }
        (x,y) \in \XX \times \XX \text{ and } \theta \in \Theta
    \end{equation*}
\end{assumption}

\begin{assumption}
    \label{a:coidp2}
    There exists a constant $\alpha > 0$ and function $K_2 \colon \XX \times \XX \to
    \RR$ such that $\int \int K_2(x, y)^2 dy \psi(\theta, x)dx < \infty$ and
    \begin{equation*}
        \| D \bar p(\theta, x, y) - D \bar p(\theta', x, y) \|_E
        \leq K_2(x,y)  \| \theta - \theta' \|^{\alpha}_E
        \quad
        \text{ for all }
        (x,y) \in \XX \times \XX \text{ and } \theta \in \Theta
    \end{equation*}
\end{assumption}

For each $\theta \in \Theta$, the pair $(p_{\theta}, g_{\theta})$
defines a covariance operator $\Sigma_{\theta}$ on $L_2$, the expression for
which is presented in \eqref{eq:defc1} below.

\begin{theorem}
    \label{t:bk3}
    If assumptions~\ref{a:le}--\ref{a:coidp2} hold and $\{X_t\}$ is
    $p_{\theta}$-Markov, then 
    \begin{equation*}
        \label{eq:far1}
        n^{-1/2} \sum_{t=1}^n \bar p(\hat \theta_n, X_t, \cdot) \tod
        N(0,\Sigma_{\theta})
        \qquad (n \to \infty)
    \end{equation*}
    and, as a consequence,
    \begin{equation}
        \label{eq:adts3}
         \hat T_n :=
          \frac{1}{n}
          \int 
            \left\{ 
                \sum_{t=1}^n \bar p(\hat \theta_n, X_t, y)  
            \right\}^2 dy 
         \; \tod \;
        \sum_{\ell = 1}^{\infty} \sigma_{\ell}^{\theta} Z_{\ell}^2 
        \qquad (n \to \infty)
    \end{equation}
    where $\{Z_{\ell}\}_{\ell \geq 1}$ is an {\sc iid} sequence of standard
    normal random variables and $\{\sigma_{\ell}^{\theta}\}_{\ell
    \geq 1}$ are the eigenvalues of $\Sigma_{\theta}$.
\end{theorem}

Let $c_{\alpha}^{\Sigma} (\theta)$ denote the $1-\alpha$ quantile of the
random variable $\sum_{\ell = 1}^{\infty} \sigma_{\ell}^{\theta} Z_{\ell}^2$.
A method for computing $c_{\alpha}^{\Sigma} (\theta)$ is given in
section~\ref{ss:epcv} below. Consider the null hypothesis
\begin{equation}
    \label{eq:h0p}
    H_0 \colon \text{ the data } \{X_t\}_{t=1}^n 
    \text{ is $p_{\theta}$-Markov for some }  \theta \in \Theta
\end{equation}
When the null is assumed to hold, we let $\theta_0 \in \Theta$ denote the true
value of $\theta$.  In view of \eqref{eq:adts3}, under the null hypothesis
\eqref{eq:h0p}, a test rejecting $H_0$ when $\hat T_n$ exceeds
$c_{\alpha}^{\Sigma}(\theta_0)$ is asymptotically of size $\alpha$.  Since
$\theta_0$ is not observable and $c_{\alpha}^{\Sigma}(\theta_0)$ cannot be
evaluated, we approximate it with $c_{\alpha}^{\Sigma}(\hat \theta_n)$.  This
gives the test
\begin{equation}
    \label{eq:test3}
    \text{ reject $H_0$  if $\hat T_n > c_{\alpha}^{\Sigma}(\hat \theta_n)$ }
\end{equation}

\begin{theorem}
    \label{t:cbk3}
    If the conditions of theorem~\ref{t:bk3} hold and $c_{\alpha}^{\Sigma}$ is
    continuous at $\theta_0$, then the test {\rm (\ref{eq:test3})} is
    asymptotically of size $\alpha$.
\end{theorem}

\subsection{Critical values}

\label{ss:epcv}

To implement the LAE test in \eqref{eq:test3}, we need a means of evaluating
$c_{\alpha}^{\Sigma}(\theta)$ for given $\theta$.  One possibility is to compute
the eigenvalues $\{\sigma_{\ell}^{\theta}\}_{\ell \geq 1}$ of $\Sigma_{\theta}$,
and then the $1-\alpha$ quantile of $\sum_{\ell = 1}^{\infty}
\sigma_{\ell}^{\theta} Z_{\ell}^2$.  A simpler method is to use simulation of the
test statistic as in algorithm~\ref{al:cc2}, which returns an approximation
to $c_{\alpha}^{\Sigma}(\theta)$.  To evaluate the critical value
$c_{\alpha}^{\Sigma}(\hat \theta_n)$ on the right-hand side of \eqref{eq:test3},
the parameter vector $\theta$ in algorithm~\ref{al:cc2} can be replaced with
the vector $\hat \theta_n$ estimated from the data.  On standard hardware and
using compiled C code based on the GNU Scientific Library, a typical calculation
with $M=2000$ and $N=1000$ completes in several seconds.\footnote{This
computation refers to a Gaussian AR(1) null hypothesis in one-dimension.  Numerical
integration was based on a 60 point fixed-order Gauss-Legendre integration
routine.} With interpreted languages the same calculation takes around 20 seconds.

\begin{algorithm}
    \vspace{0.6em}
    fix integers $M$, $N$\;
    \For{$m \in \{1,\dots,M\}$} { 
      \vspace{.2em}
      simulate a $p_{\theta}$-Markov time series $X_1^*,\ldots,X_N^*$ \;
      fit $\hat \theta_n^*$ using the simulated data set $X_1^*,\ldots,X_N^*$\;
      set $\hat T^m_N \leftarrow N^{-1}
      \int \{ \sum_{t=1}^N \bar p(\hat \theta_n^*, X_t^*, y) \}^2 dy$ \; 
    }
    {\bf return} the $1 - \alpha$ quantile of $\hat T^1_N,\ldots,\hat T^M_N$ \;
    \vspace{0.6em}
    \caption{\label{al:cc2} Approximates $c_{\alpha}^{\Sigma}(\theta)$
    corresponding to given size $\alpha$ and kernel $p_{\theta}$}
\end{algorithm}

\subsection{Consistency of the Test}

\label{ss:ct}

The LAE test in \eqref{eq:test3} is not consistent against all alternatives
in the negation of the null hypothesis specified in \eqref{eq:h0p}.  In essence, the test
compares Markov models by their stationary distribution, and models
with identical stationary distributions cannot be distinguished.  However, if
we consider the LAE test as an infinite dimension Hansen test, with null
hypothesis given in \eqref{eq:gmmfh} and alternative by
\begin{equation*}
    H_1 \colon \inf_{\theta \in \Theta}
        \| \eE \bar p(\theta, X_t, \cdot) \| > 0
\end{equation*}
then the test becomes consistent whenever the following assumptions hold:

\begin{assumption}
    \label{a:cons1}
    Under $H_1$, the sequence $\{X_t\}$ is stationary and ergodic. In
    particular, the sample mean $\frac{1}{n} \sum_{t=1}^n h(X_t)$ converges in
    probability to the expectation $\eE h(X_t)$ for all measurable $h \colon
    \XX \to L_2$ such that $\eE h(X_t)$ exists.
\end{assumption}

\begin{assumption}
    \label{a:cons3}
    The vector of partial derivatives $D p(\theta, x, y)$ exists for all $x,
    y$ in $\XX$ and all $\theta \in \Theta$.  Moreover,
    there exists a function $\eta \colon \XX \times \XX \to \RR$ such that 
    $\bE \int \eta(X_t, y)^2  dy$ is finite under $H_1$ and
    $\| D p(\theta, x, y) \|_E \leq \eta(x,y)$ for all $(x,y) \in \XX
    \times \XX$ and $\theta \in \Theta$.
\end{assumption}

\begin{assumption}
    \label{a:cons4}
    The parameter space $\Theta$ is bounded, and
    $\bE \int p(\theta, X_t, y)^2 dy$ is finite under $H_1$
    for all $\theta \in \Theta$.
\end{assumption}

\begin{theorem}
    \label{t:consthm}
    If $H_1$ is valid, $\hat \theta_n$ converges in probability  and 
    assumptions~\ref{a:cons1}--\ref{a:cons4} hold, then 
    \begin{equation*}
        \lim_{n \to \infty} 
        \bP \left\{
            \hat T_n > c_{\alpha}^{\Sigma}(\hat \theta_n)
            \right\}
            = 1
    \end{equation*}
\end{theorem}

Note
that the conditions of the theorem are sufficient but not necessary
for consistency.  While assumption~\ref{a:cons1} requires a stationary and
ergodic alternative, intuition suggests that for some choices of $H_0$ and
nonstationary alternatives, the test is likely to reject with high probability
when the sample size is large.  For example, let $H_0$ be that $\{X_t\}$ is
generated by the fixed kernel $p$ given by \eqref{eq:blv}, and consider a
random walk alternative.  (In particular, we take the same same model as $H_0$
but with $\kappa = 0$.)  The rejection probabilities for data sizes between 50
and 200 are shown in figure~\ref{f:nsc}.  By $n=200$ the rejection probability
is one.\footnote{Rejection probabilities were calculated by averaging over
1,000 observations. Since we chose a fixed kernel for $H_0$, there are no
estimated parameters, and we used algorithm~\ref{al:cc} to compute the
critical value.}

\begin{figure}
  \begin{center}
      \rotatebox{0}{\scalebox{0.6}{\includegraphics{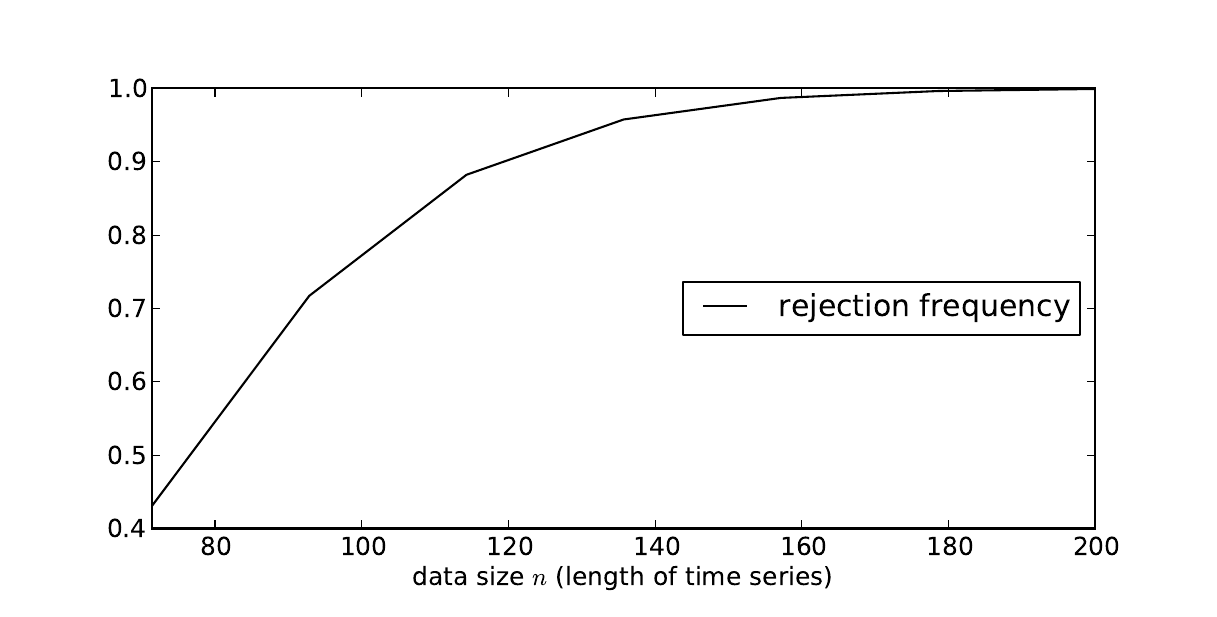}}}
  \end{center}
  \caption{Rejection frequency, nonstationary alternative \label{f:nsc}}
\end{figure}

\section{Discussion}

\label{s:d}

In this section we present applications that illustrate several features of the
test.

\subsection{Properties of the Test under $H_0$}

\label{ss:size}

In the introduction we briefly discussed the relationship between the LAE test
proposed in this paper and the test of \cite{as96}, both of which
evaluate ergodic Markov models based on $L_2$ comparison of stationary
densities.  While the LAE test compares the $L_2$ deviation between the
theoretical stationary density $\psi$ and a density estimate of the form
$\psi_n := n^{-1} \sum_{t=1}^n p(\hat \theta_n, X_t, y)$, A\"it-Sahalia's test
is based on the $L_2$ deviation between the theoretical stationary density and
a nonparametric kernel density estimate of the stationary density using the
data $X_1,\ldots, X_n$.

A\"it-Sahalia's test is a seminal contribution to the literature and his
results have initiated an important line of research.  One finding has been
that A\"it-Sahalia's test statistic might require very large data sizes to
attain its asymptotic distribution, causing excessively high rejection rates
in finite samples when the asymptotic critical value is adopted \cite{p98}.  One possible factor in slow convergence to the asymptotic distribution is
the use of nonparametric kernel density estimators in the test
statistic.\footnote{As mentioned in footnote~\ref{fn:pr}, \cite{p98} also
argues that the size of the test in finite samples is biased because the
asymptotic distribution of A\"it-Sahalia's test statistic depends only on the
stationary distribution of the process under the null.}

The test proposed here provides a new perspective on this problem.  While the
LAE test is also based on $L_2$ comparison of stationary densities, we
estimate no smoothing parameters, and use a density estimator that is
$\sqrt{n}$-consistent under the null.  These features suggest that the LAE
test might have lower size distortion in small samples.

To investigate this idea, we conduct an experiment to re-examine the
discussion of size distortion reported in \cite{p98}. Our experiment
investigates rejection rates under a true null when the sample size is
relatively small and the asymptotic critical value is used.  Following
Pritsker, the underlying model in our experiment is the Vasicek model of
interest rates.  For the DGP that generates the data $\{X_t\}$ we use the
particular Vasicek model given in \eqref{eq:blv}, while for $H_0$ we
hypothese (correctly) that the data is generated by a Vasicek model.

Beginning with A\"it-Sahalia's test, we compute the asymptotic critical value
of his test at the $5\%$ level, set $n=264$ (corresponding to 22
years of monthly observations), generate 1,000 time series of length $n$ from
the DGP, evaluate A\"it-Sahalia's test statistic for each time series, and
compare it with the asymptotic critical value.  Consistent with the results
reported in \cite{p98}, we find that A\"it-Sahalia's test rejects the
true null in over 50\% of the samples.\footnote{For details on A\"it-Sahalia's
test statistic and critical value, see \cite{as96}, p.~393.  The
bandwidth for the nonparametric kernel density estimator used in our
simulation was the optimal bandwidth for estimating the stationary density of
the Vasicek model with the true parameters.  We experimented with other
bandwidths but all choices gave a rejection rate in excess of 50\%.} On the
other hand, when we repeat the experiment with the LAE test in place of
A\"it-Sahalia's test, the LAE test rejects the true null in 4.1\% of the samples.
Thus, for this particular problem, the size distortion is largely resolved by
our test.\footnote{As with A\"it-Sahalia's test, we took $\alpha = 0.05$ and
$n=264$.  In running the experiment, we first computed the asymptotic critical
value $c^{\Sigma}_{\alpha}(\theta_0)$ for the test \eqref{eq:test3} with
$\alpha = 0.05$.  To compute this value we used algorithm~\ref{al:cc2} applied to the baseline Vasicek
density kernel with parameters given in \eqref{eq:blv} and
$M=N=5000$. Next, we simulated 1,000 times series $\{X\}_{t=1}^n$ from the
DGP, where $n=264$.  For each of these simulated time series, we used OLS to
obtain an estimate $\hat \theta_n$ for the vector of parameters of the Vasicek
model, and then used the resulting density kernel $p_{\hat \theta_n}$ to
evaluate the test statistic on the left-hand side of \eqref{eq:test3}.  Of the
1,000 time series we generated, 4.1\% of the test statistics exceeded the
asymptotic critical value.}

\subsection{Power of the Test}

\label{ss:power}

Next we investigate the power of the LAE test in finite samples.  To provide
context, we begin by re-examining a second Monte Carlo experiment of \cite{p98}, which analyzed the power of A\"it-Sahalia's test.  For
the null hypothesis he took a Vasicek model of interest rates, while for the
alternative he used the CIR model of \cite{cir85}.  He compared the
size-adjusted power of A\"it-Sahalia's test against a conditional moment-based
specification test, and found that after size adjustment, the power of
A\"it-Sahalia's test for this null-alternative pair was considerably lower than that of
the conditional moment test.  He interpreted his findings as implying that
A\"it-Sahalia's test estimates the stationary density too imprecisely to have
good power against this alternative (\cite{p98}, p.~462).

To investigate this interpretation, we now conduct a similar experiment, 
including results for the LAE test and that of the Cram\'er von Mises
test as well.  As in Pritsker, the Vasicek null associated with model
\eqref{eq:v} is paired with a discretized CIR alternative 
\begin{equation}
    \label{eq:cir}
    X_{t+1} = X_t + \kappa (b - X_t) \delta + \sigma \sqrt{X_t \delta} Z_t
    \qquad \{Z_t\} \iidsim N(0,1)
\end{equation}
As in section~\ref{ss:size}, we set $\delta=1/12$ and $n = 264$ for 22 years
of monthly observations.\footnote{In simulations we included a reflecting
barrier at zero to avoid taking the square root of negative nameyear.}  
Following Pritsker (1998, p.~460), our baseline parameter values are $\kappa =
0.89218$, $b = 0.090495$ and $\sigma =
0.180947$.  For each test we calculate the rejection frequency over 1,000
  replications.  The results of this experiment are shown in the row~1 of
  table~\ref{t:pp}.
The last four columns are the rejection frequencies at size $\alpha=0.05$ for
the Cram\'er von Mises test, a size-adjusted version of A\"it-Sahalia's test,
the LAE test with estimated parameters (test
\eqref{eq:test3}), and the conditional moment test used by Pritsker (1998,
p.~462) respectively.  As in Pritsker, for the conditional moment test
we estimate the model parameters by maximum likelihood under the Vasicek null,
and run the regression $\partial \ell(X_{t+1}, X_t)/\partial \sigma = \beta_0
+ \beta_1 X_t + u_{t+1}$, where $\ell(X_{t+1}, X_t)$ is the log likelihood of $(X_{t+1}, X_t)$ under the
null. We then conduct a two-sided test of $\beta_1 = 0$, an equality that
holds whenever the Vasicek null hypothesis is true.\footnote{\label{fn:dpc}In
the experiment, 1,000 time series of length $n$ were generated from the CIR
model with the specified parameters and monthly frequency.  In the LAE test, for
each time series, the parameters in the Vasicek null were estimated by
ordinary least squares, and the test \eqref{eq:test3} was evaluated.  In
evaluating the test, the critical value $c^{\Sigma}_{\alpha}(\hat \theta_n)$
was calculated by simulating the test statistic under the null (via
algorithm~\ref{al:cc2}, with $N=n$ and $M=500$).  Size-adjusted critical
values for the Cram\'er von Mises test and A\"it-Sahalia's test where produced
by the same method, changing only the definition of the test statistic.  Our
version of  A\"it-Sahalia's test statistic was the absolute value of the
normalized statistic (cf., e.g., Pritsker, 1998, p.~455).  The Cram\'er von
Mises test statistic was $\int (\Psi_n(y) - \Psi(\hat \theta_n, y))^2
\Psi(\hat \theta_n, dy)$ where  $\Psi(\hat \theta_n, y)$ represents the
theoretical stationary cdf under the Vasicek null at the estimated parameters,
and $\Psi_n$ denotes the empirical distribution of the data.  In the case of
A\"it-Sahalia's test, we used Silverman's rule for the bandwidth and a
standard normal Gaussian density for the kernel.}
\begin{table}
    \centering
    \begin{tabular}{cccccc}
       \toprule
       $b$ & $Z_t$ & CvM & AS & LAE & Cond m. \\
       \midrule
       0.090495 & $N(0,1)$ & 0.2487 & 0.3275 & 0.3662 & 1.0000 \\
       0.050000 & $N(0,1)$ & 0.4637 & 0.5887 & 0.6101 & 1.0000 \\
       0.025000 & $N(0,1)$ & 0.8475 & 0.9262 & 0.9375 & 1.0000 \\
       0.090495 & $t$ &  0.7825 & 0.8125 & 0.8352 & 0.5672 \\
       \bottomrule
    \end{tabular}
    \vspace{1em}
    \caption{Rejection frequency with Vasicek null and CIR alternative (test
    size 0.05)}
    \label{t:pp}
\end{table}

While Pritsker interpreted the low power of A\"it-Sahalia's test relative to
the conditional moment test as due to imprecise estimation of the stationary
density stemming from the use of nonparametric density estimators, our results
suggest that the main causes lie elsewhere.  Indeed, if we consider the first
row of table~\ref{t:pp}, we see that the LAE test offers only modest improvement
in power, despite the fact that our density estimator is $\sqrt n$-consistent.
Moreover, the empirical cdf used in the Cram\'er von Mises
test is also $\sqrt n$-consistent, and power is in fact lower than that of
A\"it-Sahalia's test.

Another possible explanation for the low power of the three tests based on the
stationary distribution (A\"it-Sahalia's test, the Cram\'er von Mises test and
the LAE test) in this setting is that the stationary distribution of the CIR
alternative with the baseline parameters can be well approximated from within
the set of Gaussian distributions available under the Vasicek null. 
In other words, once parameters are estimated, the theoretical distribution
under the null closely approximates the stationary distribution represented in
the data, relative to the natural dispersion of the test statistic
under the null.  

In order to examine this idea further, next we consider parameters in the
CIR model that produce greater skewness in the stationary distribution of the
alternative, making it more difficult to match with the Gaussian distributions
corresponding to the Vasicek null.
In rows two and three of table~\ref{t:pp}, we vary the equilibrium interest
rate $b$ from 9\% to 5\% and 2.5\% respectively, while holding $\kappa$ and
$\sigma$ fixed at the baseline values.  The skewness of the stationary density
of the CIR alternative increases as $b$ decreases.  As expected, greater
skewness leads to higher rejection rates for the three tests based on the
stationary distribution (rows 2 and 3 of table~\ref{t:pp}).

Although greater skewness improves the power of the tests based on the
stationary distribution (CvM, AS and LAE), all three remain dominated by the
conditional moment test (column 6 of table~\ref{t:pp}, rows 1--3).  However,
it is important to recall that the conditional moment test was chosen by
Pritsker precisely because of its high power against the CIR
alternative.\footnote{The conditional moment test concentrates a large amount
of its power against the CIR alternative because the expected value of the
score of the likelihood is linear under CIR.}  This is in fact the most
important reason why the power of A\"it-Sahalia's test is low relative to the
conditional moment test with a CIR alternative: This particular moment test
concentrates its power against this specific alternative, whereas
A\"it-Sahalia's test does not.  The LAE and CvM test have lower power than the
conditional moment test in this setting for exactly the same reason.

By similar reasoning, the higher power of this conditional moment test might
be fragile in practice, where the alternative is unknown.  Even if a CIR
alternative is suspected, unknown variations from the CIR alternative can
reverse the results, with the conditional moment test having lower power than
the other tests.  Row 4 of table~\ref{t:pp} illustrates this point.  Here $b$
returns to the baseline value of row 1, but the Gaussian
shock in \eqref{eq:cir} is replaced by a shock with heavy tails.\footnote{In
this simulation, the shock is $t$-distributed shock with 2.5
degrees of freedom.}  With this change, the LAE, AS and CvM tests have higher
power than the conditional moment test (table~\ref{t:pp}, row 4, columns 3--6).

\begin{table}
    \centering
    \begin{tabular}{ccccccc}
       \toprule
       $\rho$ & $\beta$ & $\gamma$ & CvM & AS & LAE & Cond m. \\
       \midrule
        0.9 & 1.0 & 0.0012 & 0.271 & 0.375 & 0.466 & 0.077 \\
        0.9 & 1.0 & 0.0034 & 0.375 & 0.653 & 0.787 & 0.123 \\
        0.9 & 1.0 & 0.0056 & 0.375 & 0.825 & 0.925 & 0.146 \\
        0.9 & 1.0 & 0.0078 & 0.348 & 0.879 & 0.958 & 0.154 \\
        0.9 & 1.0 & 0.0100 & 0.359 & 0.885 & 0.974 & 0.159 \\
       \bottomrule
    \end{tabular}
    \vspace{1em}
    \caption{Rejection frequency with Vasicek null and RSw alternative (test
    size 0.05)}
    \label{tab:rsw}
\end{table}

To further reinforce this point, next we compare the four tests with another
alternative (and the same null).  The alternative we consider is an AR(1)
model with Gaussian shocks and regime switching coefficients. The regime
switching (RSw) alternative has the form $X_{t+1} = \beta_t + \rho X_t + Z_{t+1}$,
where the shocks are standard normal, and $\{\beta_t\}$ follows a discrete
Markov process.  In particular, $\beta_0 = \beta$ where $\beta$ is a
parameter, $\beta_{t+1} = \beta_t$ with probability $1-\gamma$, and
$\beta_{t+1} = -\beta_t$ with probability $\gamma$.  Notice that when $\gamma
= 0$, this process reduces to a linear Gaussian AR(1) model.  Since the
density kernel for the linear Gaussian AR(1) shares the same parametric form
as that of the Vasicek density kernel (the density kernel $p$ in
example~\ref{eg:v}),  the null hypothesis is true when $\gamma = 0$.  Larger
values of $\gamma$ indicate greater divergence from the null.  For the other
parameters, we set $\rho = 0.9$, $\beta = 1$ and $n=500$.

The rejection frequencies for each test over 1,000 replications are shown in
table~\ref{tab:rsw} for different values of gamma.  (Other than the
alternative, the details of the calculations are the same as those described
in footnote~\ref{fn:dpc}.) Power curves for the LAE and conditional moment
tests are shown in figures~\ref{f:rsw}.  Not surprisingly, for the RSw
alternative, all three tests based on the stationary distribution (CvM, AS and
LAE) have higher power than this particular conditional moment test
(cf., table~\ref{tab:rsw}).  Of all
four, the LAE test has uniformly highest power over the values of $\gamma$ in
table~\ref{tab:rsw}.

\begin{figure}
  \begin{center}
      \rotatebox{0}{\scalebox{0.6}{\includegraphics{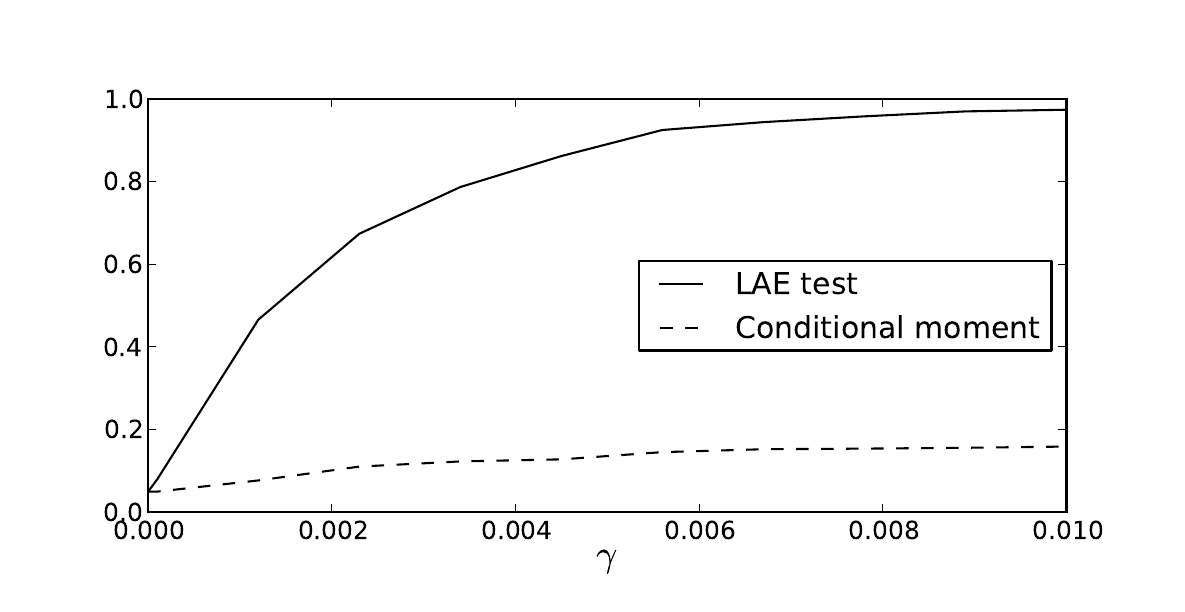}}}
  \end{center}
  \caption{Rejection frequency, regime switching alternative \label{f:rsw}}
\end{figure}

\subsection{Empirical Application}

In recent years, nonlinear business cycle models have been studied by many
authors.\footnote{See, for example, \cite{h89,pp97,ta92,hp02}.}  In this section, the LAE
test is
used to test univariate business cycle data $\{y_t\}$ for departures from a
linear Gaussian AR(1) null.  The data consists of quarterly
percentage growth rates in seasonally adjusted GDP for the US and Canada over
March 1950--September 2011.\footnote{GDP data are obtained from the IMF's
\textit{International Financial Statistics}.}
For comparison, results for the LST test of \cite{lst88}
 are also presented. The LST test is based on estimating an
AR(1) model for $y_t$ in the first stage, and then regressing the residuals
from the first stage on a polynomial in $y_{t-1}$.  The three versions of the
test we present correspond to a quadratic polynomial (LST 1), a cubic (LST
2), a quartic (LST 3), and a quartic without the cubic term (LST
4). The test statistic is $nR^{2}$, where $n$ is the sample size and $R^{2}$
  is the coefficient of determination from the second stage regression.

The $p$-values resulting from application of the LAE and LST tests to the data are
presented in table~\ref{t:bcycle}.  In the Canadian data, no tests reject the
Gaussian linear null hypothesis at either the 5\% or 10\% level.
In the US data, none of the LST tests reject the null hypothesis at either the
5\% or 10\% level, while the LAE test rejects the null hypothesis at both 
5\% and 10\%.  A likely interpretation is that the departures
from the Gaussian linear AR(1) model in the US data are other than those
featuring in the LST test.  For example, the LAE test can be sensitive to a
non-Gaussian error term, as shown in table~\ref{t:pp}.  On the other hand,
the LST test is directed only towards nonlinearities in the mean.

\begin{table}
    \centering
    \begin{tabular}{cccccc}
       \toprule
           & LAE & LST1 & LST2 & LST3 & LST4 \\
       \midrule
       Canada & 0.113 & 0.465 & 0.717 & 0.499 & 0.513 \\
       US & 0.013 & 0.672 & 0.688 & 0.800 & 0.867  \\
       \bottomrule
    \end{tabular}
    \vspace{1em}
    \caption{$p$-values for LAE and LST tests applied to rates of change in GDP}
    \label{t:bcycle}
\end{table}

\section{Conclusion}

In this paper we proposed a natural goodness of fit test for ergodic Markov
processes.  The test can be used to evaluate the hypothesis that a given time
series is generated by a parametric class of ergodic Markov models.  The power
of the test is concentrated against alternatives yielding stationary
distributions that fit poorly with the null hypothesis. No alternative needs
to be specified in order to implement the test.  Although the test is based on
a comparison of densities, the test statistic contains no smoothing
parameters, and the test has nontrivial power against $1/\sqrt{n}$ local
alternatives.  While the test is not distribution free, the critical value can
be computed consistently by Monte Carlo.  

In section~\ref{s:d} we studied the finite sample properties of our LAE test
relative to those of \cite{as96}, the Cram\'er von Mises test
and a conditional moment test designed to have power against a relatively
specific set of alternatives.  In revisiting an experiment of \cite{p98},
we found that the LAE test resolves a size distortion that occurs with
A\"it-Sahalia's test in a comparison of interest rate models.  We also
provided a new perspective on Pritsker's analysis of the power of
A\"it-Sahalia's test, and showed how goodness of fit tests based on the
stationary distribution can have good power properties relative to more
specific tests when the alternative is not completely certain.  Finally, in an empirical
application considering business cycle data, we showed how the LAE test can
detect departures from a linear AR(1) null not revealed by a more common test.

One idea not discussed above is the possibility of using weighting functions
to obtain additional power against certain alternatives.  This can be done by
adjusting the measure $\mu$ that defines the integrals in the test statistics.
(Recall from section~\ref{ss:su} that the dependence of integrals on $\mu$ is
suppressed in the notation.) The ability to apply different weighting
functions should add to the flexibility of the test.  Further investigation of
this topic is left to future research.

\section{Proofs}

\label{s:p}

In the proofs we will use the following facts without comment:  If $X$, $X_n$ and
$Y_n$ are $L_2$ random elements with $X_n \tod X$ in $L_2$ and $\|X_n
- Y_n\| = o_P(1)$ in $\RR$, then $Y_n \tod X$ (cf., e.g., \cite{d02},
lemma~11.9.4).  If $\alpha_n = o_P(1)$ in $\RR$ and $f \in L_2$, then $X_n :=
\alpha_n f$ is an $L_2$-valued random variable with $\| X_n \| = o_P(1)$ in
$\RR$.  If $X_n$ is an $L_2$-valued random variable, then the statement $X_n =
o_P(1)$ means that $\|X_n\| = o_P(1)$ in $\RR$.
The statement that $G \sim N(m, C)$ on $\hH$ is equivalent to the statement
$\bE \exp( i \la h, G \ra) = \exp\{ i \la h, m \ra - \la h, C h \ra / 2 \}$
for all $h \in \hH$, where $i$ is the imaginary unit (cf., e.g.,
\cite{p67}, theorem 6.4.), yielding the characterization
\begin{equation}
    \label{eq:ecf}
    G \sim N(m, C) \text{ on } \hH
    \; \iff \;
    \la G, h \ra \sim N( \la h, m \ra,  \la h, C h \ra )
    \text{ on $\RR$ for all } 
    h \in \hH
\end{equation}
Hence, the distribution of $G$ is defined by the values $\la h, m \ra$
and $\la h, C h \ra$ over $h \in \hH$.  

First we turn to the proof of lemma~\ref{l:dsn}.  In the proof we specialize to
the case $\hH = L_2$, which is sufficient for our purposes (and changes
nothing except notation).  As background to the proof,
note first that since any covariance operator $C$ is linear, positive, symmetric and Hilbert-Schmidt
on $L_2$ (see, for example, Bosq, 2000, theorem~1.7), we can
apply the spectral theorem in $L_2$ to obtain the decomposition 
\begin{equation}
    \label{eq:decomp}
    C h = \sum_{\ell = 1}^{\infty} \lambda_{\ell} \la h, v_{\ell} \ra v_{\ell}
    \qquad (h \in L_2)
\end{equation}
where $(v_{\ell})_{\ell \geq 1}$ is an orthonormal basis of $L_2$ consisting
of eigenfunctions of $C$, and $(\lambda_{\ell})_{\ell \geq 1}$ is the
corresponding eigenvalues (i.e., $C v_{\ell} = \lambda_{\ell}
v_{\ell}$ for all $\ell$).  The eigenvalues are real, nonnegative, and satisfy
$\sum_{\ell \geq 1} \lambda_{\ell} < \infty$. 

\begin{proof}[Proof of lemma~\ref{l:dsn}]
    As a first step, let us show that if $C$ is a covariance operator,
    $\{(\lambda_{\ell}, v_{\ell}\}_{\ell \geq 1}$ are its eigenelements and
    $\{Z_{\ell}\}$ are {\sc iid} and standard normal, then
    \begin{equation}
        \label{eq:crg}
        G_0 := \sum_{\ell = 1}^{\infty} \lambda_{\ell}^{1/2} Z_{\ell} v_{\ell}
        \sim N(0, C)
    \end{equation}
    To show that $G_0$ is zero-mean Gaussian on $L_2$, we must show
    that $\la G_0, h \ra$ is zero-mean Gaussian in $\RR$ for every $h \in L_2$.
    It suffices to show that this property holds on an orthonormal subset of
    $L_2$.\footnote{Let $(e_{\ell})$ be any orthonormal subset of $L_2$, and
    suppose that $\la G_0, e_{\ell} \ra$ is zero-mean Gaussian in $\RR$ for each
    $\ell \in \NN$.  Pick any $h \in L_2$.  Then $\la G_0, h \ra =
    \sum_{\ell} \la G_0, e_{\ell} \ra \la h, e_{\ell} \ra$.  The right-hand side
    is the almost sure limit of zero mean Gaussians, and hence is itself zero-mean
    Gaussian.}  Choosing $(v_{\ell})$ as our orthonormal subset, we have 
    \begin{equation*}
        \la G_0, v_{\ell} \ra 
        = \sum_k \lambda_k^{1/2} Z_k \la v_k, v_{\ell} \ra
        =  \lambda_{\ell}^{1/2} Z_{\ell} \sim N(0, \lambda_{\ell})
    \end{equation*}
    We also need to prove that $C$
    is the covariance operator of $G_0$, or $\bE \la G_0, g \ra \la G_0, h \ra = \la
    g, C h \ra$ for any $g, h \in L_2$.  It suffices to show the same for
    $h, g \in \{v_{\ell}\}_{\ell \in \NN}$.  Fixing $j, k \in \NN$, we have
    \begin{equation*}
        \bE \la G_0, v_j \ra \la G_0, v_k \ra 
        = \bE \lambda_j^{1/2} Z_j \lambda_k^{1/2} Z_k 
        = \lambda_k \1\{j = k\}
    \end{equation*}
    On the other hand, since $v_k$ is an eigenfunction of $C$ with
    eigenvalue $\lambda_k$, we have $\la v_j , C v_k \ra = \la v_j,
    \lambda_k v_k \ra = \lambda_k \1\{j = k\}$.  Hence
    $C$ is the covariance operator of $G_0$, and $G_0 \sim N(0, C)$ as
    claimed.

    Returning to the proof of lemma~\ref{l:dsn}, let $G \sim N(0, C)$ on
    $L_2$, and consider the distribution of $\|G\|^2$.  We
    have now shown that $G$ and $G_0$ have the same distribution, so $\|G\|^2$
    and $\|G_0\|^2$ also share the same distribution. To complete the proof of
    the lemma, observe that Parseval's
    identity gives
        $\| G_0 \|^2 
        = \sum_{\ell=1}^{\infty} \la G_0, v_{\ell} \ra^2
        = \sum_{\ell=1}^{\infty} \lambda_{\ell} Z_{\ell}^2$.
\end{proof}

In what follows, we make repeated use of the following Markov Hilbert space central
limit theorem, which is a simple corollary of \cite{s06}, theorem~3.1. 
In the statement of the theorem, $\{X_t\}$ is a stationary Markov process on
$\XX$, the function $F_0 \colon \XX \to L_2$ is Borel measurable, and $F := F_0 - \eE F_0(X_t)$. 

\begin{theorem}
    \label{t:hclt}
    If $\{X_t\}$ is geometrically ergodic and, for the function $V$ in
    \eqref{eq:vuedef}, there exists nonnegative constants
    $c_0$, $c_1$ and $\gamma$ such that $\gamma < 1$ and $\| F_0(x) \|^2 \leq
    c_0 + c_1 V(x)^{\gamma}$ for all $x \in S$,
    then $n^{-1/2} \sum_{t=1}^n F(X_t)$ converges to a
    centered Gaussian on $L_2$, with covariance operator $C$
    satisfying
    \begin{equation}
        \label{eq:adefc}
        \la h, C h \ra 
        = \bE \la F(X_1), h \ra^2
        +  2 \sum_{t=2}^{\infty} \bE \la F(X_1), h \ra \la F(X_t), h \ra
    \end{equation}
\end{theorem}

The following results will also be needed in the proofs below:

\begin{lemma}
    \label{l:moco}
    Let $p$ be a density kernel, and let $\psi$ be its stationary
    density. If $p$ is $V$-mixing, then $\psi \in L_2$,
    $p(x,\cdot) \in L_2$ and $\bar p(x, \cdot) \in L_2$ for all $x \in \XX$.
    Moreover, if $X$ is any $\XX$-valued random variable, then  $y \mapsto
    p(X,y)$ is an $L_2$-valued random variable.\footnote{The statement that
    $p(X,\cdot)$ is an $L_2$-valued random variable includes the
    claim that $\Omega \ni \omega \mapsto p(X(\omega),\cdot) \in L_2$ is 
    measurable.  See the proof for details.}
\end{lemma}

\begin{proof}
    Evidently (\ref{eq:bop}) implies that $p(x,\cdot) \in L_2$ for each $x \in
    \XX$. Regarding the claim that $\psi \in L_2$, the definition of
    stationarity and Jensen's inequality give
    \begin{equation*}
        \int \psi(y)^2 dy
        = \int \left[ \int p(x,y) \psi(x) dx \right]^2 dy 
        \leq \int \int p(x,y)^2 \psi(x) dx \, dy 
    \end{equation*}
    From \eqref{eq:bop} and (\ref{eq:vuedef}), we then have
    \begin{equation*}
        \int \psi(y)^2 dy
         \leq \int \int p(x,y)^2  dy \psi(x) dx 
         \leq c_0 + c_1 \int V(x)^{\gamma} \psi(x) dx 
         < \infty
    \end{equation*}
    Since $\gamma < 1$ we can apply Jensen's inequality to obtain
    \begin{equation*}
        \int V(x)^{\gamma} \psi(x) dx
        \leq \left[ \int V(x) \psi(x) dx \right]^{\gamma}
    \end{equation*}
    and this expression is finite by \eqref{eq:vuedef}.  We conclude that
    $\psi \in L_2$ as claimed.  Moreover,
    we can now see that $\bar p(x,\cdot) \in L_2$ for any $x \in \XX$, because
    \begin{equation*}
        \| \bar p(x, \cdot) \| 
        = \| p(x,\cdot) - \psi(\cdot) \| \leq  \| p(x,\cdot) \| + \| \psi \|
    \end{equation*}
    To show that
    $p(X,\cdot)$ is an $L_2$-valued random variable, we need to prove that
    $\Omega  \ni \omega \mapsto p(X(\omega),\cdot) \in L_2$ is also measurable, in the
    sense that preimages of Borel subsets of $L_2$ are measurable in $\Omega$.
    Since $L_2$ is separable, it follows from the Pettis measurability theorem
    that any mapping $\Omega \ni \omega \mapsto g(\omega) \in L_2$ is
    measurable whenever $\Omega \ni \omega \mapsto \la g(\omega), h \ra \in
    \RR$ is measurable for each $h \in L_2$.  Using this fact, the
    measurability of $\omega \mapsto p(X(\omega),\cdot)$ is easily verified.
    This concludes the proof of lemma~\ref{l:moco}.
\end{proof}

\begin{lemma}
    \label{l:ptbk2}
    If $p$ is $V$-mixing and $\{X_t\}$ is $p$-Markov, then $\eE \bar
    p(X_t,\cdot) = 0$ for all $t$.
\end{lemma}

\begin{proof}
    Fixing $t$ and letting $X = X_t$, this amounts to the claim that, for any
    $h \in L_2$ we have
    \begin{equation*}
        \bE \int \bar p(X,y) h(y) dy = 0
    \end{equation*}
    Fix $h \in L_2$.  Note that for each $y \in \XX$ we have
    \begin{equation}
        \label{eq:ebarpy}
        \bE \bar p(X,y) 
        = \int p(x,y) \psi(x)dx - \psi(y) = \psi(y) - \psi(y)  = 0
    \end{equation}
    As a consequence, $\bE \int \bar p(X,y) h(y) dy 
    = \int \bE \bar p(X,y) h(y) dy = 0$ whenever Fubini's theorem is
    valid. Fubini's theorem is valid whenever $\bE \int | \bar p(X,y) h(y) |
    dy < \infty$.  To check this, observe that, by the
    Cauchy-Schwartz and triangle inequalities,
    \begin{equation*}
        \int | \bar p(x,y) h(y) | dy 
         \leq \| \bar p(x,\cdot)  \| \| h \| 
         \leq (\| p(x,\cdot) \|  + \| \psi \| ) \| h \| 
    \end{equation*}
    Hence it suffices to show that $\bE \| p(X,\cdot) \|^2 = \int \int
    p(x,y)^2 dy \psi(x) dx < \infty$.  This
    claim was verified as part of the proof of lemma~\ref{l:moco}.
\end{proof}

\begin{lemma}
    \label{l:msl}
    Let $\{X_t\}$ be a geometrically ergodic Markov process on $\XX$ and let
    $V$ be the weight function in \eqref{eq:vuedef}. If $\hat X_t := (X_t,
    \ldots, X_{t+r})$, then $\{\hat X_t\}$ is Markov and
    geometrically ergodic on $\XX^{r+1}$ with weight function $\hat V
    (x_0,\ldots,x_r) := \sum_{k=0}^r V(x_k)$.  
\end{lemma}

\begin{proof}
    To see that $\{ \hat X_t \}$ is Markov, pick any bounded measurable $h
    \colon \XX^{r+1} \to \RR$.  We have
    \begin{align*}
        \bE [ h(\hat X_t) \given \hat X_{t-1}, \ldots, \hat X_1]
        & = \bE [ h(X_t, \ldots, X_{t+r}) \given X_{t+r-1}, \ldots, X_1] \\
        & = \bE [ h(X_t, \ldots, X_{t+r}) \given X_{t+r-1}, \ldots, X_{t-1}] 
    \end{align*}
    where the second equality is due to the Markov property of $\{X_t\}$.
    Changing notation we can write this as
    \begin{equation*}
        \bE [ h(\hat X_t) \given \hat X_{t-1}, \ldots, \hat X_1]
        = \bE [ h(\hat X_t) \given \hat X_{t-1}]
    \end{equation*}
    Since $h$ is an arbitrary bounded measurable function, we have shown that
    $\{\hat X_t\}$ is Markov as claimed.
    
    Next consider geometric ergodicity.  Let $V$ and $\hat V$ be as in the
    statement of the lemma.
    Pick any measurable $h \colon \XX^{r+1} \to \RR$ such that $|h| \leq 1$.  
    Fix $(x_0,\ldots, x_r) \in \XX^{r+1}$.  Let $\psi_t := p^{t-r}(x_r,
    \cdot)$, which is the density of $X_t$ given $X_r = x_r$.  Observing that 
    \begin{equation*}
        \hat \psi(x_t,\ldots,x_{t+r}) := \psi(x_t) p(x_t, x_{t+1}) \cdots
            p(x_{t+r-1}, x_{t+r})
    \end{equation*}
    is the stationary density of $\hat X_t$ and
    \begin{equation*}
        \hat \psi_t(x_t,\ldots,x_{t+r}) := \psi_t(x_t) p(x_t, x_{t+1}) \cdots
            p(x_{t+r-1}, x_{t+r})
    \end{equation*}
    is the density of $\hat X_t$ given $\hat X_0 = (x_0, \ldots, x_r)$, we then have
    \begin{equation}
        \label{eq:handg}
        \left| \int h \hat \psi_t - \int h \hat \psi \right|
            = \left| \int g \psi_t - \int g \psi \right|
    \end{equation}
    for $g(x_t) := \int \cdots \int h(x_t,\ldots,x_{t+r}) p(x_t,x_{t+1})
    \cdots p(x_{t+r-1}, x_{t+r}) dx_{t+1} \cdots dx_{t+r}$.
    Given that $|h| \leq 1$, we also have $|g| \leq 1$, and therefore
    \begin{equation*}
        \left| \int g \psi_t - \int g \psi \right|
        \leq 2 \sup_{B \in \xX} 
            \left| \int_B  \psi_t - \int_B \psi \right|
        = 2 \sup_{B \in \xX} 
        \left| \int_B  p^{t-r}(x_r,y) dy  - \int_B \psi(y) dy \right|
    \end{equation*}
    From this bound, \eqref{eq:handg} and \eqref{eq:vuedef}, we then have
    \begin{equation*}
        \left| \int h \hat \psi_t - \int h \hat \psi \right|
        \leq 2 \lambda^{t-r} L V(x_r)
        \leq 2 \lambda^{t-r} L  \sum_{k=0}^r V(x_k)
        = \lambda^t \left( \frac{2L}{\lambda^r} \right) \hat V(x_0,\ldots,x_r)
    \end{equation*}
    Specializing to $h = \1_B$ and recalling \eqref{eq:vuedef}, we see that
    $\{\hat X_t\}$ satisfies the condition on the right-hand side of
    \eqref{eq:vuedef}.  Regarding the finiteness condition on the left-hand
    side of \eqref{eq:vuedef}, observe that $\bE \hat V(\hat X_t) =
    \sum_{k=0}^r \bE V(X_{t+r}) = r \int V d \psi$.  This term is finite by
    geometric ergodicity of the original process $\{X_t\}$.  We conclude that
    $\{\hat X_t\}$ is geometrically ergodic with weight function $\hat V$.  
\end{proof}

\subsection{Theorems~\ref{t:bk2} and \ref{t:dms}}

\begin{proof}[Proof of theorem~\ref{t:bk2}]
    Let $p$ be $V$-mixing and let $\{X_t\}_{t=1}^n$ be $p$-Markov.  Define $F_0(X_t)
    := p(X_t, \cdot)$ and let $F(X_t) := \bar p(X_t, \cdot) = p(X_t, \cdot) -
    \psi$.  We saw in lemmas~\ref{l:moco} and \ref{l:ptbk2} that
    $F_0(X_t)$ is an $L_2$-valued random variable satisfying
    $\eE F_0(X_t) = \psi$. Moreover, $\| F_0(x) \|^2 \leq c_0 + c_1
    V(x)^{\gamma}$ for all $x \in \XX$ by \eqref{eq:bop}.  Applying
    theorem~\ref{t:hclt}, we then have the weak convergence $n^{-1/2}
    \sum_{t=1}^n F(X_t) \tod N(0,C)$, where $C$ is defined in
    \eqref{eq:adefc}.  It is straightforward to check that this expression and
    \eqref{eq:defc0} are identical, and hence $C = \Lambda$. In summary,
    $n^{-1/2} \sum_{t=1}^n \bar p(X_t,\cdot) \tod N(0,\Lambda)$ as claimed.
\end{proof}

\begin{proof}[Proof of theorem~\ref{t:dms}]
    Let $\{X_t\}_{t=1}^n$ be $p$-Markov.  Define
    \begin{equation}
        \label{eq:lae}
        \psi_n(y) := \frac{1}{n} \sum_{t=1}^n p(X_t, y) 
        \quad \text{and} \quad
        \psi_{kn}'(y) := \frac{1}{n} \sum_{t=1}^{kn} p(X'_t, y)
    \end{equation}
    Let $(U_{\ell})_{\ell \geq 1}$ and $(U_{\ell}')_{\ell \geq 1}$ be
    mutually independent {\sc iid} sequences of standard normal random
    variables.  Fix $k \in \NN$, and consider the decomposition
    \begin{equation*}
        n^{1/2} (\psi_n - \psi_{kn}') 
         = n^{1/2} (\psi_n - \psi) - k^{-1/2} (kn)^{1/2} (\psi_{kn}' - \psi) 
    \end{equation*}
    Note that $n^{1/2} (\psi_n - \psi)$ and $(kn)^{1/2} (\psi_{kn}' - \psi)$
    are independent random functions in $L_2$.  By theorem~\ref{t:bk2} and the
    representation \eqref{eq:crg} we have
    \begin{equation*}
        n^{1/2} (\psi_n - \psi) \tod \sum_{\ell} \lambda_{\ell}^{1/2} U_{\ell} v_{\ell}
        \quad \text{and} \quad
        (kn)^{1/2} (\psi_{kn}' - \psi) 
        \tod \sum_{\ell} \lambda_{\ell}^{1/2} U_{\ell}' v_{\ell}
    \end{equation*}
    By independence and continuity of addition and scalar multiplication in
    $L_2$, we then have
    \begin{equation*}
        n^{1/2} (\psi_n - \psi'_{nk})
        \tod 
        \sum_{\ell} \lambda_{\ell}^{1/2} U_{\ell} - 
            k^{-1/2} \sum_{\ell} \lambda_{\ell}^{1/2} U_{\ell}' v_{\ell}
         = \sum_{\ell} \lambda_{\ell}^{1/2} (U_{\ell} - k^{-1/2}U_{\ell}') v_{\ell}
    \end{equation*}
    Applying the continuous mapping theorem and the Pythagorean law, we obtain
    \begin{equation*}
        n \| \psi_n - \psi'_{nk} \|^2
        \tod 
        \| \sum_{\ell} \lambda_{\ell}^{1/2} 
            (U_{\ell} - k^{-1/2}U_{\ell}') v_{\ell} \|^2
        =  \sum_{\ell} \lambda_{\ell} (U_{\ell} - k^{-1/2}U_{\ell}')^2 
    \end{equation*}
    The left-hand side of this equation is equal to the left-hand side
    of (\ref{eq:dmsn}).  Moreover, if $Z_{\ell}$ is standard normal, then $(1
    + 1/k) Z_{\ell}^2$ and $(U_{\ell} - k^{-1/2}U_{\ell}')^2$ have the same
    law.  This completes the
    proof of (\ref{eq:dmsn}).
\end{proof}

\subsection{Local Alternatives: Theorem~\ref{t:la}}

Let $\hH$ be defined as $L_2 \times \RR$, with inner product
\begin{equation*}
    \la g, h \ra = \la g_1, h_1 \ra + g_2 h_2
    \qquad (g = (g_1, g_2) \text{ and } h = (h_1, h_2) )
\end{equation*}
(Here $\la g, h \ra$ is the inner product in $\hH$ and $\la g_1, h_1 \ra$ is
the inner product in $L_2$.  The notation does not distinguish between them,
but the meaning will be clear from context.)
With the norm $\| h \| = \sqrt{\la h, h \ra}$, the space $\hH$ is a Hilbert
space, and the norm topology of $\hH$ corresponds to the product topology of
$L_2 \times \RR$.
The next result is an extension of the Cram\'er-Wold theorem to $\hH$:

\begin{lemma}
    \label{l:acih}
    Let $U_n := (Y_n, \ell_n)$ be a random sequence in $\hH$, where $Y_n$ is a
    random element of $L_2$ and $\ell_n$ is a random variable for all $n$.
    Let $U$ be a Gaussian random element of $\hH$ with distribution $N(m, S)$.
    \begin{equation*}
        \la U_n, h \ra \tod N( \la h, m \ra, \la h, S h \ra)
        \text{ in } \RR \text{ for all } h \in \hH
        \implies
        U_n \tod U \text{ in } \hH
    \end{equation*}
\end{lemma}

\begin{proof}
    Suppose for the moment that $\{ U_n \}$ is tight in $\hH$.  In this case,
    to show that $U_n$ converges in distribution to $U$ in $\hH$, we need only
    show that $\la U_n, h \ra$ converges in distribution to $\la U, h \ra$ in
    $\RR$ for all $h \in \hH$ (Bosq, 2000, theorem~2.3).  This is immediate
    from \eqref{eq:ecf},
    which tells us that $\la U, h \ra$ has distribution $N( \la h, m \ra, \la
    h, S h \ra)$.

    It remains to show that $\{ U_n \}$ is tight in $\hH$.  To see that this
    is so, note that, as required in the lemma, $\la U_n, h \ra$ converges in
    distribution for all $h \in \hH$.  Choosing $h = (0, 1)$, we see that
    $\ell_n$ converges in distribution, and is therefore tight.
    Now fix $\epsilon > 0$.  Since $\{Y_n\}$ and $\{ \ell_n \}$ are both
    tight, we can find compact sets $K_a \subset L_2$ and $K_b \subset \RR$
    with $\bP \{ Y_n \notin K_a \} < \epsilon / 2$ and $\bP \{ \ell_n \notin
    K_b \} < \epsilon / 2$ for all $n$.  The set $K_a \times K_b$ is compact
    in the product topology on $\hH$, and we have
    \begin{equation*}
        \bP \{ U_n \notin K_a \times K_b \}
         \leq \bP \{ Y_n \notin K_a \} \cup \{ \ell_n \notin K_b \} 
         \leq \bP \{ Y_n \notin K_a \} + \bP \{ \ell_n \notin K_b \} 
        < \epsilon
    \end{equation*}
    We conclude that $\{U_n\}$ is tight in $\hH$, completing the proof of
    lemma~\ref{l:acih}.
\end{proof}

\begin{lemma}
    \label{l:joco}
    Let $Y_n$, $\bQ_n$ and $\bP_n$ be as defined in section~\ref{ss:la}, let
    \begin{equation*}
        r(x,y) := \frac{k(x,y)}{p(x,y)},
        \qquad 
        \sigma^2 
        := \bE r(X_t, X_{t+1})^2 
        = \int r(X_t,X_{t+1})^2 d \bP_n
    \end{equation*}
    and let $\ell_n \colon \XX^n \to \RR$ be the log likelihood ratio
    \begin{equation*}
        \ell_n 
        = \log \frac{d \bQ_n}{d \bP_n} 
        = \log \left\{
            \frac{ \prod_{t=2}^n p_n(X_{t-1}, X_t) }
                { \prod_{t=2}^n p(X_{t-1}, X_t) }
              \right\}
    \end{equation*}
    If $U_n := (Y_n, \ell_n)$ and $U = (Y, \ell)$ has distribution $N(m, S)$
    for $m = (0, -\sigma^2/2)$ and $S$ satisfying
    \begin{equation*}
        \la h, S h \ra 
            = \la h_1, \Lambda h_1 \ra + 2 \la \tau, h_1 \ra h_2 + h_2^2 \sigma^2
            \qquad (h = (h_1, h_2) \in \hH)
    \end{equation*}
    then $U_n \tod U$ under $\bP_n$.
\end{lemma}

\begin{proof}
    In what follows, all probabilities and expectations are evaluated under
    $\bP_n$.  We begin by obtaining a more convenient expression for the likelihood
    ratio $\ell_n$.  Writing $p_t$ for $p(X_{t-1}, X_t)$ and 
    $k_t$ for $k(X_{t-1}, X_t)$, we have
    \begin{equation*}
        \ell_n = \sum_{t=2}^n \{ \log(p_t + k_t / \sqrt{n}) - \log(p_t) \}
    \end{equation*}
    Expanding the $\log$ function around $p_t$ yields
    \begin{equation*}
        \ell_n 
        = \frac{1}{\sqrt{n}} \sum_{t=2}^n \frac{k_t}{p_t}
        - \frac{1}{2 n} \sum_{t=2}^n \frac{k_t^2}{p_t^2}
        +  \frac{1}{3 n^{3/2}} \sum_{t=2}^n 
            \frac{ k_t^3 }
                { [p_t + \lambda n^{-1/2} \, k_t]^3 }
    \end{equation*}
    For some $\lambda \in [0,1]$.  Since $(X_{t-1}, X_t)$ is itself ergodic
    (see lemma~\ref{l:msl}) and
    \begin{equation*}
        \frac{ k_t^3 }
            { [p_t + \lambda n^{-1/2} \, k_t]^3 }
        \leq 
        \sup_{\delta \in [0,1]} 
        \frac{ |k(X_{t-1}, X_t)|^3 }{ |p(X_{t-1}, X_t) + \delta \, k(X_{t-1}, X_t)|^3 }
    \end{equation*}
    it follows from assumption~\ref{a:laa} that
        $\frac{1}{n} \sum_{t=2}^n
            k_t^3 / [p_t + \lambda n^{-1/2} \, k_t]^3  = O_P(1)$,
    and hence
    \begin{equation}
        \label{eq:elln}
        \ell_n 
        = \frac{1}{\sqrt{n}} \sum_{t=2}^n r(X_{t-1}, X_t)
        - \frac{1}{2 n} \sum_{t=2}^n r(X_{t-1}, X_t)^2
        + o_P(1)
    \end{equation}
    Now we return to the proof of lemma~\ref{l:joco}.  Taking into account
    lemma~\ref{l:acih} and the fact that $\{Y_n\}$ is tight in $L_2$---as
    implied by the convergence in theorem~\ref{t:bk2}---it suffices to show
    that
    \begin{equation}
        \label{eq:sfdp}
        \la U_n, h \ra
        \tod
        N( - h_2 \sigma^2 / 2, 
            \la h_1, \Lambda h_1 \ra + 2 \la \tau, h_1 \ra h_2 + h_2^2 \sigma^2)
    \end{equation}
    for arbitrary $h \in \hH$.  Fixing such an $h = (h_1, h_2)$, the
    definition of $U_n$ and our expression for $\ell_n$ in \eqref{eq:elln}
    gives
    \begin{equation*}
        \la U_n, h \ra
        = \frac{1}{\sqrt n} \sum_{t=1}^n \la h_1, \bar p(X_t, \cdot) \ra
        + h_2 \frac{1}{\sqrt{n}} \sum_{t=2}^n r(X_{t-1}, X_t)
        - h_2 \frac{1}{2 n} \sum_{t=2}^n r(X_{t-1}, X_t)^2
        + o_P(1)
    \end{equation*}
    Since $(X_{t-1}, X_t)$ is ergodic (lemma~\ref{l:msl}) we have
    \begin{equation*}
        \frac{h_2}{2} \frac{1}{n} \sum_{t=2}^n r(X_{t-1}, X_t)^2
        \to h_2 \frac{\sigma^2}{2}
        \quad \text{in probability}
    \end{equation*}
    As a result of this convergence and Slutsky's theorem, the result
    \eqref{eq:sfdp} will be confirmed if we show that
    \begin{equation}
        \label{eq:cq}
        \frac{1}{\sqrt n} \sum_{t=2}^n q(X_{t-1}, X_t)
        \tod
        N(0, \la h_1, \Lambda h_1 \ra + 2 \la \tau, h_1 \ra h_2 + h_2^2 \sigma^2)
    \end{equation}
    for $q(X_{t-1}, X_t) := \la h_1, \bar p(X_t, \cdot) \ra + h_2 r(X_{t-1}, X_t)$.
    To see that this is indeed the case, observe first that $\bE \la h_1, \bar
    p(X_t, \cdot) \ra = 0$ under $H_0$ by lemma~\ref{l:ptbk2}. Also,
    \begin{equation*}
        \bE r(X_{t-1}, X_t)
        = \int \int \frac{k(x,y)}{p(x,y)}  \psi(x) p(x,y) dx dy
        = \int \left[ \int k(x,y) dy \right]  \psi(x) dx 
        = 0
    \end{equation*}
    It follows that $\bE q(X_{t-1}, X_t) = 0$, and, as a result of
    geometric ergodicity of $(X_{t-1}, X_t)$ (lemma~\ref{l:msl}) and the
    scalar CLT for geometrically ergodic Markov processes (e.g., \cite{mt09},
    theorem~17.0.1), we have
    \begin{equation*}
        \frac{1}{\sqrt n} \sum_{t=2}^n q(X_{t-1}, X_t)
        \tod
        N(0, v),
        \quad
        v := \bE q(X_1, X_2)^2 
            + 2 \sum_{t=2}^{\infty} \bE q(X_1, X_2) q(X_t, X_{t+1})
    \end{equation*}
    It remains only to show that $v =  \la h_1, \Lambda h_1 \ra + 2 \la \tau, h_1
    \ra h_2 + h_2^2 \sigma^2$, which is the right-hand side of the variance in
    \eqref{eq:cq}.  Prior to proving this, we observe that all of the
    following statements are valid, and will be used without comment below:
    \begin{itemize}
        \item $\bE r(X_1, X_2) = 0$ and $\bE [ r(X_t, X_{t+1}) \given X_1] =
            0$ for all $t \geq 1$.
        \item $\bE [ r(X_1, X_2)r(X_t, X_{t+1}) ] = 0$ for all $t \geq 2$.
        \item $\bE [ \la h_1, \bar p(X_1, \cdot) \ra r(X_t, X_{t+1}) = 0$ 
            for all $t \geq 1$.
    \end{itemize}
    (The first of these statements has already been established above, and the
    proofs of the rest are similar.)  Turning now to the evaluation of $v$,
    note that
    \begin{align*}
        \bE q(X_1, X_2)^2 
        & = \bE \{ \la h_1, \bar p(X_1, \cdot) \ra^2 
            + 2 \la h_1, \bar p(X_1, \cdot) \ra h_2 r(X_1, X_2)
            +  h_2^2 r(X_1, X_2)^2 \}
            \\
        & = \bE  \la h_1, \bar p(X_1, \cdot) \ra^2  +  h_2^2 \sigma^2 
    \end{align*}
    while, for any given $t \geq 2$,
    \begin{equation*}
        \bE q(X_1, X_2)q(X_t, X_{t+1}) 
        = \bE  \la h_1, \bar p(X_1, \cdot) \ra \la h_1, \bar p(X_t, \cdot) \ra  
            +  \bE \la h_1, \bar p(X_t, \cdot) h_2 r(X_1, X_2)
    \end{equation*}
    As a result, we have
    \begin{multline*}
        v = 
        \bE  \la h_1, \bar p(X_1, \cdot) \ra^2  
        + 2 \sum_{t=2}^{\infty} 
            \bE  \la h_1, \bar p(X_1, \cdot) \ra \la h_1, \bar p(X_t, \cdot) \ra  
            \\
        + 2 \sum_{t=2}^{\infty}
            \bE \la h_1, \bar p(X_t, \cdot) h_2 r(X_1, X_2)
        +  h_2^2 \sigma^2
    \end{multline*}
    In view of \eqref{eq:defc0}, we have
    \begin{equation*}
        \bE  \la h_1, \bar p(X_1, \cdot) \ra^2  
        + 2 \sum_{t=2}^{\infty} 
            \bE  \la h_1, \bar p(X_1, \cdot) \ra \la h_1, \bar p(X_t, \cdot) \ra  
        = \la h_1, \Lambda h_1 \ra
    \end{equation*}
    Finally, using the definition of $\tau$, we obtain 
    $v =  \la h_1, \Lambda h_1 \ra + 2 \la \tau, h_1 \ra h_2 + h_2^2 \sigma^2$.
    This verifies \eqref{eq:cq}, and completes the proof of
    lemma~\ref{l:joco}.
\end{proof}

\begin{lemma}
    \label{l:aell}
    For $\ell_n, \ell$ defined in lemma~\ref{l:joco}, we have $\ell_n
    \tod \ell$ and $\bE \exp(\ell) = 1$ under $\bP_n$.
\end{lemma}

\begin{proof}
    We saw in lemma~\ref{l:joco} that, under $\bP_n$, we have $\la h, U_n \ra
    \tod \la h, U \ra$ for all $h \in \hH$, where $U \sim N(m, S)$ for $m$ and
    $S$ defined in lemma~\ref{l:joco}.  Specializing to $h = (0,1)$ obtains
    the first claim in lemma~\ref{l:aell}.
    Regarding the second claim in lemma~\ref{l:aell}, for this same $h$ we
    have $\ell = \la h, U \ra = N( \la m, h \ra, \la h, S h \ra)$, and
    given the definitions of $m$ and $S$ in
    lemma~\ref{l:joco}, 
    \begin{equation*}
        N( \la m, h \ra, \la h, S h \ra)
        = N( - \sigma^2/2, \sigma^2)
    \end{equation*}
    \begin{equation*}
        \fore
        \bE \exp(\ell) 
        = \exp \left( -\frac{\sigma^2}{2} + \frac{\sigma^2}{2} \right)
        = 1
    \end{equation*}
    when expectation is taken under $\bP_n$.  This completes the proof.
\end{proof}

We are now ready to complete the proof of theorem~\ref{t:la}.

\begin{proof}[Proof of theorem~\ref{t:la}]
    We saw in lemmas~\ref{l:joco} and \ref{l:aell} that if $\ell_n = d \bQ_n /
    d \bP_n$ is the log likelihood ratio, then under $\bP_n$ we have 
    \begin{equation*}
        \begin{pmatrix}
            Y_n \\
            \ell_n
        \end{pmatrix}
        \tod 
        \begin{pmatrix}
            Y \\
            \ell
        \end{pmatrix}
        \sim N(m, S)
    \end{equation*}
    and, moreover, $\bE \exp(\ell) = 1$.  Applying the abstract version of Le
    Cam's third lemma presented in \cite{vw96},
    theorem~3.10.7, we then have
    \begin{equation*}
        Y_n \tod \pi \text{ under } \bQ_n
    \end{equation*}
    when $\pi$ is the probability measure on $L_2$ defined by
    \begin{equation*}
        \pi(f) = \bE \exp(\ell) f(Y)
        \quad \text{for all bounded measurable } \;
        f \colon \XX \to \RR
    \end{equation*}
    To complete the proof of theorem~\ref{t:la}, we need only show that
    \begin{equation}
        \label{eq:pif}
        \pi := N(\tau, \Lambda)
    \end{equation}
    To see that this equality holds, let $V$ be a random element on $L_2$ with
    $V \sim \pi$.  In view of \eqref{eq:ecf}, to verify \eqref{eq:pif} it
    sufficies to show that, for arbitrary fixed $a \in L_2$, we have
    \begin{equation}
        \label{eq:dv}
        \la a, V \ra \sim N( \la a, \tau \ra, \la a, \Lambda a \ra)
    \end{equation}
    To establish \eqref{eq:dv}, observe that, from the definition of $\pi$,
    the moment generating function of $\la a, V \ra$ is 
    \begin{equation*}
        M(t) 
        := \bE \exp(t \la a, V \ra)
        = \bE \exp(\ell) \exp(t \la a, Y \ra)
        = \bE \exp(t \la a, Y \ra + \ell)
    \end{equation*}
    If $h \in \hH$ is defined as $h := (ta, 1)$ and $U := (Y, \ell)$, then
    $\la h, U \ra$ is precisely $t \la a, Y \ra + \ell$.  Since $U$ is
    Gaussian, we know that $\la h, U \ra$ is Gaussian, and therefore $t \la a,
    Y \ra + \ell$ is Gaussian in $\RR$.  Its expectation and variance are
    given by
    \begin{equation*}
        \bE (t \la a, Y \ra + \ell) 
        = \bE \la (ta, 1), U \ra
        = \la (ta, 1), m \ra
        = \la (ta, 1), (0, -\sigma^2/2) \ra
        = - \sigma^2/2
    \end{equation*}
    and
    \begin{equation*}
        \Var \la (ta, 1), U \ra
        = \la (ta, 1) , S (ta, 1) \ra
        = t^2 \la a, \Lambda a \ra + 2 t \la \tau, a \ra + \sigma^2
    \end{equation*}
    where the final expression follows from the definition of $S$ given in the
    statement of lemma~\ref{l:joco}.
    To finish the proof, we observe that, since 
    $t \la a, Y \ra + \ell$ is Gaussian with mean and variance as derived
    above, we must have
    \begin{equation*}
        \bE \exp(t \la a, Y \ra + \ell)
        = \exp \left\{ 
            - \frac{\sigma^2}{2}  
            + t \la \tau, a \ra 
            + t^2 \frac{\la a, \Lambda a \ra}{2}
            + \frac{\sigma^2}{2}
            \right\}
    \end{equation*}
    Cancelling the two instances of $\sigma^2/2$, we find that the moment
    generating function of $\la a, V \ra$ is
    \begin{equation*}
        M(t)
        = \exp \left\{ 
            t \la \tau, a \ra 
            + t^2 \frac{\la a, \Lambda a \ra}{2}
            \right\}
        \qquad (t \in \RR)
    \end{equation*}
    This is precisely the moment generating function for the $N( \la a, \tau
    \ra, \la a, \Lambda a \ra)$ distribution, and hence we have established
    \eqref{eq:dv}.  This completes the proof of theorem~\ref{t:la}.
\end{proof}

\subsection{Theorems~\ref{t:bk3} and \ref{t:cbk3}}

We begin by defining $\Sigma_{\theta}$.  Given $\theta \in \Theta$ and
$p_{\theta}$-Markov sequence $\{X_t\}$, we let $\Sigma_{\theta}$ be the
operator defined by
\begin{equation}
    \label{eq:defc1}
    \la h, \Sigma_{\theta} h \ra 
    = \la h, \Lambda_{\theta} h \ra + 2 \bE P_1 Q_1 + \bE Q_1^2
    + 2 \sum_{t = 2}^{\infty} \bE  
            \left\{ P_1 Q_t + Q_1 P_t + Q_1 Q_t \right\}
\end{equation}
for $P_j := \int \bar p(\theta, X_j, y) h(y) \, dy$ and $Q_j := \int \bE \{ D
\bar p(\theta, X_1, y) \}^{\top} g(X_j,\ldots,X_{j+r}) h(y) \, dy$.  Here
$\Lambda_{\theta}$ is the operator (\ref{eq:defc0}) corresponding to
$p_{\theta}$.

\begin{proof}[Proof of theorem~\ref{t:bk3}]
    Assume the conditions of theorem~\ref{t:bk3}.  Fix $\theta \in \Theta$ and
    let $\{X_t\}$ be $p_{\theta}$-Markov.  We need to prove the
    statement
    \begin{equation}
        \label{eq:bk3c}
        \hat Y_n := n^{-1/2} \sum_{t=1}^n \bar p(\hat \theta_n, X_t, \cdot) 
        \tod N(0, \Sigma_{\theta})
    \end{equation}
    in $L_2$.  Throughout the proof, we use the notation $\rho(y) := \bE D
    \bar p(\theta, X_t, y)$ and $\rho_m(y) := \bE D_m \bar p(\theta, X_t, y)$.
    By differentiability (assumption~\ref{a:coid0}) we can expand $p$
    around $\theta$ to get
    \begin{equation}
        \label{eq:bbp}
        \bar p(\hat \theta_n, x, y)   
        = \bar p(\theta, x, y) 
            + D \bar p(\theta, x, y)^{\top} (\hat \theta_n - \theta)
            + R(\hat \theta_n, x, y)
    \end{equation}
    where $R$ is the remainder term and $\top$ indicates inner product in
    $\RR^M$.  We then have
    \begin{equation*}
        \hat Y_n(y)
        = n^{-1/2} \sum_{t=1}^n 
            \left\{
                \bar p(\theta, X_t, y) +
                D \bar p(\theta, X_t, y)^{\top} (\hat \theta_n - \theta)
                + R(\hat \theta_n, X_t, y) 
            \right\}
    \end{equation*}
    Adding and subtracting $\rho(y)^{\top}(\hat \theta_n - \theta)$, we can
    write this last expression as
    \begin{equation}
        \label{eq:cdpop}
        \hat Y_n(y)
        = n^{-1/2} \sum_{t=1}^n 
            \left\{
                \bar p(\theta, X_t, y) +
                \rho(y)^{\top} (\hat \theta_n - \theta)
            \right\}
            + I_n(y) + J_n(y)
    \end{equation}
    where $I_n := n^{-1/2} \sum_{t=1}^n [D \bar p(\theta, X_t, \cdot) -
    \rho]^{\top} (\hat \theta_n - \theta)$ and $J_n := n^{-1/2} \sum_{t=1}^n
    R(\hat \theta_n, X_t, \cdot)$.  As a first step of the proof, we show that
    $I_n = J_n = o_P(1)$ in $L_2$.  Beginning with
    $I_n$, observe that
    \begin{equation}
        \label{eq:inn}
        \| I_n \|
        = 
            \sum_{m=1}^M
            |\hat \theta_n^m -\theta^m|
        \left\| 
            n^{-1/2} \sum_{t=1}^n 
            \{ D_m \bar p(\theta, X_t, \cdot) - \rho_m \}
        \right\|
    \end{equation}
    Fix $m \in \{1, \ldots, M\}$.  An application of the definition of
    $\eE$ verifies that $\eE D_m \bar p(\theta, X_t,
    \cdot) = \rho_m$.  Moreover, assumption~\ref{a:coid0} gives
    \begin{equation*}
        \| D_m \bar p(\theta, x, \cdot) \|^2
        = \int  D_m \bar p(\theta, x, y)^2 dy
        \leq V_{\theta}(x)^{1/2}
    \end{equation*}
    As a result, theorem~\ref{t:hclt} applies, and hence
    $n^{-1/2} \sum_{t=1}^n \{ D_m \bar p(\theta, X_t, \cdot) -
    \rho_m\}$ converges in distribution to a centered Gaussian in $L_2$.
    Applying the continuous mapping theorem, the norm of this random function
    also converges in distribution, and hence
    \begin{equation*}
        \left\| 
             n^{-1/2} \sum_{t=1}^n 
             \left\{
                 D_m \bar p(\theta, X_t, \cdot) - \rho_m
             \right\}
        \right\|
        = O_P(1)
    \end{equation*}
    Since $|\hat \theta_n^m -\theta^m| = o_P(1)$ by assumption, we then have
    \begin{equation*}
        |\hat \theta_n^m -\theta^m|
        \left\| 
             n^{-1/2} \sum_{t=1}^n 
             \left\{
                 D_m \bar p(\theta, X_t, \cdot) - \rho_m
             \right\}
        \right\|
        = o_P(1) O_P(1)  = o_P(1)
    \end{equation*}
    for each $m \in \{1, \ldots, M\}$.  Returning to \eqref{eq:inn} we see
    that $I_n = o_P(1)$ as claimed.  
    
    Turning to the case of $J_n$, we claim that
    \begin{equation}
        \label{eq:remainder}
        \| J_n \| 
        =
        \left\| 
            n^{-1/2} \sum_{t=1}^{n} R(\hat \theta_n, X_t, \cdot) 
        \right\| 
        = o_{p}(1)
    \end{equation}
    Using the mean value theorem, we can write
    \begin{equation*}
        R(\hat \theta_n, X_t, y)
        = \{D \bar p(\tilde \theta,X_{t},y)
        - D\bar p(\theta,X_{t},y)\}^{\top}(\hat \theta_n -\theta)
    \end{equation*}
    where $\tilde \theta$ lies on the line segment between $\theta$ and
    $\hat \theta_n$.  It follows that
    \begin{equation*}
         n^{-1/2} \sum_{t=1}^n R(\hat \theta_n, X_t, y) 
        = 
        \left[
        \frac{1}{n}\sum_{t=1}^{n}
        \{D \bar p(\tilde \theta,X_{t},y) - D\bar p(\theta,X_{t},y)\}
        \right]^{\top}
        n^{1/2} (\hat \theta_n -\theta)
    \end{equation*}
    Applying the Cauchy-Schwartz inequality in $\RR^M$, we obtain
    \begin{equation}
        \label{eq:rocs}
        \left| n^{-1/2} \sum_{t=1}^n R(\hat \theta_n, X_t, y) \right|
        \leq
        H_n(y) \, n^{1/2} \| \hat \theta_n -\theta \|_E
    \end{equation}
    where $\| \cdot \|_E$ is the norm in $\RR^M$, and
    \begin{equation*}
        H_n(y) := 
        \left\|
        \frac{1}{n}\sum_{t=1}^{n}
        \{D \bar p(\tilde \theta,X_{t},y) - D\bar p(\theta,X_{t},y)\}
        \right\|_E
    \end{equation*}
    From (\ref{eq:rocs}) we obtain the $L_2$ norm inequality
    \begin{equation*}
        \left\| n^{-1/2} \sum_{t=1}^n R(\hat \theta_n, X_t, \cdot) \right\|
        \leq \| H_n \| \cdot O_P(1)
    \end{equation*}
    Hence, to establish (\ref{eq:remainder}), it suffices to prove that $\|
    H_n \| = o_P(1)$.  By the definition of $H_n$ and
    assumption~\ref{a:coidp2}, we have
    \begin{align*}
        \| H_n \| 
        & \leq \frac{1}{n} \sum_{t=1}^n 
        \left[
        \int 
        \| D \bar p(\tilde \theta,X_{t},y) 
            - D\bar p(\theta,X_{t},y) \|_E^2 \, dy
        \right]^{1/2}
        \\
        & \leq 
            \| \tilde \theta - \theta \|_E^{\alpha}
            \frac{1}{n} \sum_{t=1}^n 
                \left[
                    \int K_2(X_t, y)^2 \, dy   
                \right]^{1/2}
    \end{align*}
    By assumption~\ref{a:le}, $\| \hat \theta_n - \theta \|_E^{\alpha} =
    o_P(1)$.  Moreover, by Jensen's inequality,
    \begin{equation*}
        \bE \,
            \left[
                \int K_2(X_t, y)^2 \, dy   
            \right]^{1/2}
        \leq 
            \left[
                \bE
                \int K_2(X_t, y)^2 \, dy   
            \right]^{1/2}
        =
            \left[
                \int \int K_2(x, y)^2 \, dy \, \psi_{\theta}(x) dx
            \right]^{1/2}
    \end{equation*}
    This expression is finite by assumption~\ref{a:coidp2}.
    Applying the scalar law of large nameyear for ergodic
    Markov processes (e.g., \cite{mt09}, theorem~17.1.7), we have 
    \begin{equation*}
        \frac{1}{n} \sum_{t=1}^n 
            \left[
                \int K_2(X_t, y)^2 \, dy   
            \right]^{1/2}
            = O_P(1)
    \end{equation*}
    We conclude that $\| H_n \| \leq o_P(1)O_P(1)  = o_P(1)$, and hence (\ref{eq:remainder}) is
    valid.  

    Returning now to \eqref{eq:cdpop}, we have shown that the last two terms
    on the right-hand side are $o_P(1)$, while assumption~\ref{a:le} and
    simple manipulations show that the first term can be expressed as 
    \begin{equation*}
        n^{-1/2} \sum_{t=1}^n 
            \left\{
                \bar p(\theta, X_t, y) +
                \rho(y)^{\top} g_{\theta}(X_t,\ldots,X_{t+r}) 
            \right\}
            + o_P(1)
    \end{equation*}
    Define $M_t := (X_t,\ldots, X_{t+r})$,
    \begin{equation*}
        F_0(M_t) := p(\theta, X_t, \cdot) + \rho(\cdot)^{\top} g_{\theta}(M_t)
        \quad \text{and} \quad
            F(M_t) := \bar p(\theta, X_t, \cdot) + \rho(\cdot)^{\top}
            g_{\theta}(M_t)
    \end{equation*}
    We see that \eqref{eq:bk3c} will be established if we can show that
    \begin{equation}
        \label{eq:bk3c2}
        n^{-1/2} \sum_{t=1}^n F(M_t) 
        :=
        n^{-1/2} \sum_{t=1}^n 
            \left\{
                \bar p(\theta, X_t, \cdot) +
                \rho^{\top} g_{\theta}(M_t) 
            \right\}
            \tod N(0, \Sigma_{\theta})
    \end{equation}
    We will use theorem~\ref{t:hclt}.  As a first step, we claim that $\eE
    F_0(M_t) = \psi$.  Since $F(M_t) = F_0(M_t) - \psi$, it suffices to show
    that $\eE F(M_t) = 0$. To see that this is so, pick any $h \in L_2$.  From the
    definition and Fubini's theorem we have
    \begin{align*}
        \bE \la F(M_t), h \ra
        & = \bE \int \bar p(\theta, X_t, y) h(y) dy +
        \bE \int \rho(y)^{\top} g_{\theta}(M_t) h(y) dy \\
        & = \int \bE \bar p(\theta, X_t, y) h(y) dy +
        \int \rho(y)^{\top} \bE[g_{\theta}(M_t)] h(y) dy 
    \end{align*}
    Since $\{X_t\}$ is $p_{\theta}$-Markov, both of these expectations are
    zero (by the definition of $\bar p$ and assumption~\ref{a:le}
    respectively), and hence $\eE F(M_t) = 0$ as claimed.

    Let $\hat V(x_0,\ldots,x_r) := \sum_{k=0}^r V(x_k)$.  We saw in
    lemma~\ref{l:msl} that $\{M_t\}$ is geometrically ergodic with weight
    function $\hat V$. In order to apply theorem~\ref{t:hclt}, it remains to
    show that there exists constants $c_0$, $c_1$, $\gamma$ with $\gamma < 1$
    and
    \begin{equation}
        \label{eq:bntac}
        \| F_0(x_0, \ldots, x_r) \|^2 
            \leq c_0 + c_1 \hat V(x_0, \ldots x_r)^{\gamma}
        \qquad \text{for all }  (x_0, \ldots, x_r) \in \XX^{r+1}
    \end{equation}
    To establish \eqref{eq:bntac}, observe first that
    \begin{align*}
        \| F_0(x_0, \ldots, x_r) \|^2
        & = \| p(\theta, x_0, \cdot) + \rho(\cdot)^{\top} g_{\theta}(x_0, \ldots, x_r) \|^2
           \\
           & \leq 2 \int p(\theta, x_0, y)^2 dy + 
           2 \int [ \rho(y)^{\top} g_{\theta}(x_0, \ldots, x_r) ]^2 dy
            \\
           & \leq 2 \int p(\theta, x_0, y)^2 dy + 2 \int \| \rho(y) \|_E^2 dy \;
           \| g_{\theta}(x_0, \ldots, x_r) \|_E^2
    \end{align*}
    Note that $\int \| \rho(y) \|_E^2 dy$ is finite.  Indeed, using Jensen's
    inequality and assumption~\ref{a:coid0}, we have 
    \begin{align*}
      \| \rho(y) \|_E^2 dy  
      & = \sum_{m=1}^M \int \{ \bE D_m \bar p(\theta, X_t, y) \}^2 dy
      \\
      & \leq \sum_{m=1}^M \bE \int \{ D_m \bar p(\theta, X_t, y) \}^2 dy
      \\
      & \leq \sum_{m=1}^M \bE (V_{\theta}(X_t)^{1/2})
           \leq \sum_{m=1}^M (\bE V_{\theta}(X_t))^{1/2}
    \end{align*}
    The final expression is finite by \eqref{eq:vuedef}, and hence $\int \| \rho(y)
    \|_E^2 dy$ is finite as claimed.  As a result, combining
    \eqref{eq:bop} and assumption~\ref{a:le}, there are
    nonnegative constants $c_0$, $a_1$, $a_2$ and $\alpha < 1$ with
    \begin{align*}
        \| F_0(x_0, \ldots, x_r) \|^2
          & \leq c_0 + a_1 V(x_0)^{\alpha} + 
                a_2 \hat V(x_0, \ldots x_k)^{2/(2+\delta)}
            \\
          & \leq c_0 + a_1 \hat V(x_0, \ldots x_k)^{\alpha} + 
                a_2 \hat V(x_0, \ldots x_k)^{2/(2+\delta)}
    \end{align*}
    Setting $\gamma := \max\{\alpha, 2/(2+\delta)\}$ and $c_1 := \max\{a_1,
    a_2\}$ yields \eqref{eq:bntac}.
    The conditions of theorem~\ref{t:hclt} are now verified, and from that
    theorem we obtain $n^{-1/2} \sum_{t=1}^n F(M_t) \tod N(0, S)$ with 
    \begin{equation}
        \label{eq:L}
        \la h, S h \ra 
        = \bE \la F(M_1), h \ra^2
        +  2 \sum_{t=2}^{\infty} \bE \la F(M_1), h \ra \la F(M_t), h \ra
    \end{equation}
    for arbitrary $h \in L_2$.
    Thus \eqref{eq:bk3c2} will be established if we can show that 
    $\la h, S h \ra = \la h, \Sigma_{\theta} h \ra$, which is to say
    that the
    right-hand side of \eqref{eq:L} agrees with the right-hand side of
    \eqref{eq:defc1}. Observe that $\la F(M_j), h \ra = P_j + Q_j$, where
    $P_j$ and $Q_j$ are defined immediately after \eqref{eq:defc1}.  As a
    result, we can write
    \begin{equation*}
        \la h, S h \ra 
        = \bE P_1^2 + 2 \bE P_1 Q_1 + \bE Q_1^2
        +  2 \sum_{t=2}^{\infty} \bE \left\{ P_1 P_t + P_1 Q_t + Q_1 P_t + Q_1 Q_t \right\}
    \end{equation*}
    Since $\la h, \Lambda_{\theta} h \ra = \bE P_1^2 + 2 \sum_{t=2}^{\infty} \bE
    P_1 P_t$ it follows that $\la h, S h \ra = \la h, \Sigma_{\theta} h
    \ra$ as claimed.  Hence \eqref{eq:bk3c2} is valid, completing the proof of
    theorem~\ref{t:bk3}.
\end{proof}

\begin{proof}[Proof of theorem~\ref{t:cbk3}]
    The claim in the theorem is that under $H_0$ we have
    \begin{equation}
        \label{eq:cicbk3}
        \lim_{n \to \infty} 
        \bP \left\{ \hat T_n \leq c^{\Sigma}_{\alpha}(\hat \theta_n) \right\}
        \geq 1 - \alpha
    \end{equation}
    If $H_0$ holds, then $\hat T_n \tod \sum_{\ell} \sigma_{\ell}(\theta_0) Z_{\ell}^2$ and
    $c^{\Sigma}_{\alpha}(\hat \theta_n) \toprob c^{\Sigma}_{\alpha}(\theta_0)$,
    where the first result is due to \eqref{eq:adts3} and the the second
    is due to consistency of $\hat \theta_n$ and continuity of
    $c^{\Sigma}_{\alpha}$ at $\theta_0$.  Slutsky's theorem yields
    $\hat T_n - c^{\Sigma}_{\alpha}(\hat \theta_n) + c^{\Sigma}_{\alpha}(\theta_0)
        \tod \sum_{\ell} \sigma_{\ell}(\theta_0) Z_{\ell}^2$.
    As a result,
    \begin{align*}
        \lim_{n \to \infty} 
        \bP\{ \hat T_n \leq c^{\Sigma}_{\alpha}(\hat \theta_n) \} 
        & = 
        \lim_{n \to \infty} 
        \bP\{\hat T_n - c^{\Sigma}_{\alpha}(\hat \theta_n) + c^{\Sigma}_{\alpha}(\theta_0) 
        \leq c^{\Sigma}_{\alpha}(\theta_0) \}
        \\
        & = \bP \left\{ \sum_{\ell} \sigma_{\ell}(\theta_0) Z_{\ell}^2
        \leq c^{\Sigma}_{\alpha}(\theta_0) \right\}
    \end{align*}
    By the definition of $c^{\Sigma}_{\alpha}(\theta_0)$ this probability is $1-\alpha$.
\end{proof}

\subsection{Consistency: Theorem~\ref{t:consthm}}

\begin{proof}[Proof of theorem~\ref{t:consthm}]
    Assume the conditions of the theorem.  Let $\theta_1$ be the probability
    limit of $\hat \theta_n$.  Recalling that the test statistic is $n \|
    n^{-1} \sum_{t=1}^{n}\bar{p}(\hat{\theta}_{n},X_{t},\cdot)\|^2$ and
    observing that
    \begin{equation*}
      \left\| 
            \frac{1}{n} \sum_{t=1}^n \bar p(\hat \theta_n, X_t,\cdot)
                - \eE \bar p(\hat \theta_n, X_t, \cdot)
      \right\|
            + \left\| 
            \frac{1}{n}\sum_{t=1}^n \bar{p}(\hat{\theta}_{n},X_{t},\cdot)
              \right\|
      \geq \| \eE \bar p(\hat \theta_n, X_t,\cdot)\|
      \geq \epsilon
    \end{equation*}
    where $\epsilon > 0$ is the value of the infimum in the definition of
    $H_1$, we see that the claim in the theorem will be valid whenever
    \begin{equation}
        \label{eq:sfcons}
        \left\| 
            \frac{1}{n} \sum_{t=1}^n \bar p(\hat \theta_n, X_t,\cdot)
                - \eE \bar p(\hat \theta_n, X_t, \cdot)
        \right\|
        =
        \left\| 
            \frac{1}{n} \sum_{t=1}^n  p(\hat \theta_n, X_t,\cdot)
                - \eE p(\hat \theta_n, X_t, \cdot)
        \right\|
    \end{equation}
    converges in probability to zero as $n \to \infty$.  The term in
    \eqref{eq:sfcons} is bounded above by $(I) + (II) + (III)$ where
    \begin{equation*}
        (I) := 
        \left\| 
            \frac{1}{n} \sum_{t=1}^n \{ p(\hat \theta_n, X_t,\cdot)
            - p(\theta_1, X_t, \cdot) \}
        \right\|, \quad
        (II) := 
        \left\| 
            \frac{1}{n} \sum_{t=1}^n p(\theta_1, X_t,\cdot)
                - \eE p(\theta_1, X_t, \cdot)
        \right\|
    \end{equation*}
    and $(III) := \left\| \eE p(\theta_1, X_t, \cdot) - \eE p(\hat \theta_n,
    X_t, \cdot) \right\|$.  Here $\theta_1$ is the in probability limit of
    $\hat \theta_n$.  We claim that all of these terms converge to zero.
    To begin, consider first the term $(I)$.  Evidently
    \begin{equation*}
        (I) \leq
            \frac{1}{n} \sum_{t=1}^n
        \left\| 
             p(\hat \theta_n, X_t,\cdot) - p(\theta_1, X_t, \cdot) 
        \right\|
    \end{equation*}
    Fix $y \in \XX$.  By the mean value theorem in $\RR^M$, there exists a $\bar \theta \in
    \RR^M$ such that
    \begin{equation*}
        | p(\hat \theta_n, X_t, y) - p(\theta_1, X_t, y) |
            \leq \| Dp(\bar \theta, X_t, y) \|_E \cdot 
                \| \hat \theta_n - \theta_1 \|_E
    \end{equation*}
    Applying assumption~\ref{a:cons3} now yields
    \begin{equation}
        \label{eq:bmvt}
        | p(\hat \theta_n, X_t, y) - p(\theta_1, X_t, y) |
            \leq \eta(X_t, y) \| \hat \theta_n - \theta_1 \|_E
            \qquad \text{for all } y \in \XX
    \end{equation}
    Taking the the $L_2$ norm of this expression and averaging over $t$, we obtain
    \begin{equation*}
        (I) \leq 
            \| \hat \theta_n - \theta_1 \|_E
            \;
            \frac{1}{n} \sum_{t=1}^n
            h(X_t)
            \quad \text{where} \quad
            h(x)
            :=
            \left[ 
                \int \eta(x, y)^2 \, dy 
            \right]^{1/2}
    \end{equation*}
    By Jensen's inequality and assumption~\ref{a:cons3},
    \begin{equation*}
        \bE h(X_t)
        = \bE
            \left[ 
                \int \eta(X_t, y)^2 \, dy 
            \right]^{1/2}
        \leq
            \left[ \bE \int \eta(X_t,y)^2 \, dy \right]^{1/2} < \infty
    \end{equation*}
    Since this expectation is finite
    and $\{X_t\}$ is assumed to be stationary and ergodic, 
    assumption~\ref{a:cons1} implies that
    $n^{-1} \sum_{t=1}^n h(X_t)$ converges to $\bE h(X_t)$ in probability.
    Hence we have $(I) \leq o_P(1) O_P(1) = o_p(1)$ as claimed.
    
    Turning to the term $(II)$, the claim that this is $o_P(1)$ follows
    directly from assumption~\ref{a:cons1}, provided that the expectation
    $\eE p(\theta_1, X_t, \cdot)$ exists.  This $L_2$ expectation exists
    whenever the scalar expectation of the norm of $p(\theta_1, X_t, \cdot)$
    is finite.  Finiteness of this scalar expectation is a direct consequence
    of assumption~\ref{a:cons4}.

    Regarding $(III)$, another application of \eqref{eq:bmvt}  gives
    \begin{equation*}
        (III)
         \leq 
            \bE \, \| p(\theta_1, X_t, \cdot) - p(\hat \theta_n, X_t, \cdot) \|
         \leq \bE \,
            \| \hat \theta_n - \theta_1 \|_E
            \;
            \left[ 
                \int \eta(X_t,y)^2 \, dy 
            \right]^{1/2}
    \end{equation*}
    Using the Cauchy-Schwartz inequality, we obtain
    \begin{equation*}
        (III)
        \leq
            \left[ 
            \bE \,
            \| \hat \theta_n - \theta_1 \|_E^2
            \;
            \bE \, \int \eta(X_t,y)^2 \, dy 
            \right]^{1/2}
    \end{equation*}
    The constant term $\bE \, \int \eta(X_t,y)^2 \, dy$ is finite by
    assumption~\ref{a:cons3}.  Moreover $\| \hat \theta_n - \theta_1 \|_E =
    o_P(1)$ implies $\| \hat \theta_n - \theta_1 \|_E^2 = o_P(1)$.  Moreover,
    $\| \hat \theta_n - \theta_1 \|_E^2$
    is uniformly bounded as a result of the boundedness of $\Theta$.
    Hence $\bE \, \| \hat \theta_n - \theta_1 \|_E^2$ converges to zero.

    We conclude that $(I) + (II) + (III) = o_P(1) + o_P(1) + o(1) = o_P(1)$.
    Hence the term in \eqref{eq:sfcons} is $o_P(1)$ and the proof is done.
\end{proof}

\end{document}